\documentclass[aps,pra,twocolumn,amsmath,amssymb,superscriptaddress,floatfix]{revtex4-1}
\usepackage{color}
\usepackage{amsmath,amsfonts,amssymb}
\usepackage{amsthm}
\usepackage{bm}\usepackage{graphicx}
\usepackage{algorithm}
\usepackage{bbm}
\usepackage{epsfig}
\usepackage{verbatim}
\usepackage{array}
\usepackage{physics}

\usepackage{ulem}
\usepackage[T1]{fontenc}
\usepackage{tgtermes}

\newcommand*{\black}{\textcolor{black}} 

\theoremstyle{plain}
\newtheorem{thm}{Theorem}

\theoremstyle{definition}

\DeclareMathOperator{\cov}{\mathrm{cov}}

\DeclareMathOperator{\cfim}{\mathcal{F}_{\mathrm{cl}}}
\DeclareMathOperator{\cfiminv}{\mathcal{F}_{\mathrm{cl}}^{-1}}
\DeclareMathOperator{\trcfiminv}{\tr\mathcal{F}_{\mathrm{cl}}^{-1}}

\usepackage[unicode=true,bookmarks=true,bookmarksnumbered=false,bookmarksopen=false, breaklinks=false,pdfborder={0 0 1},backref=false,colorlinks=true]{hyperref}
\hypersetup{linkcolor=magenta, urlcolor=blue, citecolor=blue, pdfstartview={FitH}, unicode=true}

\begin{document}

\title{Generalizable control for multiparameter quantum metrology}
\author{Han Xu}
\affiliation{Department of Physics, City University of Hong Kong, Tat Chee Avenue, Kowloon, Hong Kong SAR, China and Shenzhen Research Institute, City University of Hong Kong, Shenzhen, Guangdong 518057, China}
\affiliation{School of Physics and Technology, Wuhan University, Wuhan 430072, China}
\author{Lingna Wang}
\affiliation{Department of Mechanical and Automation Engineering, Chinese University of Hong Kong, Shatin, Hong Kong SAR, China}
\author{Haidong Yuan}
\email{hdyuan@mae.cuhk.edu.hk}
\affiliation{Department of Mechanical and Automation Engineering, Chinese University of Hong Kong, Shatin, Hong Kong SAR, China}
\author{Xin Wang}
\email{x.wang@cityu.edu.hk}
\affiliation{Department of Physics, City University of Hong Kong, Tat Chee Avenue, Kowloon, Hong Kong SAR, China and Shenzhen Research Institute, City University of Hong Kong, Shenzhen, Guangdong 518057, China}

\begin{abstract}
    Quantum control can be employed in quantum metrology to improve the precision limit for the estimation of unknown parameters. The optimal control, however, typically depends on the actual values of the parameters \black{and thus} needs to be designed adaptively with the updated estimations of those parameters. Traditional methods, such as gradient ascent pulse engineering (GRAPE), need to be rerun for each new set of parameters encountered, making the optimization costly\black{,} especially when many parameters are involved. Here we study the generalizability of optimal control, namely\black{,} optimal controls that can be systematically updated across a range of parameters with minimal cost. In cases where control channels can completely reverse the shift in the Hamiltonian due to \black{a} change in parameters, we provide an analytical method which efficiently generates optimal controls for any parameter starting from an initial optimal control found \black{by either} GRAPE or reinforcement learning. When the control channels are restricted, the analytical scheme is invalid\black{,} but reinforcement learning still retains a level of generalizability, albeit in a narrower range. \black{In} cases where the shift in the Hamiltonian is impossible to \black{decompose} to available control channels, no generalizability is found for either the reinforcement learning or the analytical scheme. We argue that the generalization of reinforcement learning is through a mechanism similar to the analytical scheme. Our results provide insights \black{into} when and how the optimal control in \black{multiparameter} quantum metrology can be generalized, thereby facilitating efficient implementation of optimal quantum estimation of multiple parameters, particularly for an ensemble of systems with ranges of parameters.
\end{abstract}


\maketitle

\section{Introduction}
One of the main goals of quantum metrology is to identify the highest possible precision for the estimation of unknown parameters \black{and} find practical ways to approach that limit.
There have been extensive studies on the optimization of probe states and measurements to achieve better precisions \cite{Humphreys2013Quantum,pezze2017optimal,yue2014quantum,yuan2017quantum}. 
Recently, it \black{was} realized that quantum control during the evolution can be optimized to improve the precision in the estimation of a single parameter \cite{yuan2015Optimal,liu2017quantum,pang2017optimal,fraisse2017enhancing,fiderer2019maximal,haine2020machine} or multiple parameters \cite{yuan2016sequential,liu2017control}. Numerical methods such as \black{gradient} ascent pulse engineering (GRAPE) have been employed to obtain optimal controls  \cite{Khaneja2005Optimal,liu2017quantum,liu2017control}. The optimal controls, however, typically depend on the values of the parameters, which need to be adaptively updated with the accumulation of measurement data. In this process, one typically starts with a nominal (or assumed) value of the unknown parameter, \black{finds} the control optimal for that nominal value, and performs a measurement using that optimized control. This measurement gives an estimated value of the unknown parameter, along with an accuracy which should fall in \black{a} certain theoretical limit dependent on the control process. \black{Then, if} one is not satisfied with this theoretical limit, the process is repeated  with a new control optimized at the last estimated value until \black{a} certain convergence \black{criterion} is met. For an ensemble of systems with parameters to be estimated, varying in a range, such optimization has to be repeated many times, making it \black{resource consuming} and cumbersome \black{[see Fig.~\ref{fig:fig1}(a)]}. A method that can systematically update optimal controls across a range of parameters with minimal cost, a characteristic termed \black{\emph{generalizability}}, is then desirable, particularly for the estimation of multiple parameters as the search space becomes much larger, and performing a brand-new search of optimal controls for each set of  parameters could become prohibitively expensive.

Reinforcement \black{learning} (RL) is an alternative method that can be used to estimate multiple parameters. The main idea of RL is to convert an optimization problem into a \black{game}, \black{with} the goals and constraints of the problem \black{encapsulated} as rules in the game. A computer agent, through playing the game in an emulator, obtains rewards depending on its performance. A procedure that maximizes rewards therefore leads to an optimized solution of the desired problem. \black{Reinforcement learning} has demonstrated its superior power in playing several types of games \cite{mnih2013playing,mnih2015human,silver2016mastering}. It has also been shown to improve solutions to many problems in quantum information science \cite{Thomas2018Reinforcement}, quantum control \cite{Bukov2018Reinforcement,niu2019universal,zhang2019does,lin2020quantum}, and in \black{particular quantum} metrology \cite{xu2019generalizable,schuff2020improving,fiderer2020neural}. In a previous work on the single-parameter quantum estimation, we have shown that RL has an important advantage: the generalizability \cite{xu2019generalizable}. An agent trained for one value of the parameter can be straightforwardly generalized to other values of the parameter, producing near-optimal controls \black{[see Fig.~\ref{fig:fig1}(b)]}. This makes RL efficient for quantum estimation with an ensemble of systems where the values of the parameter may vary within a certain range.

\begin{figure}
    \includegraphics[width=0.8\linewidth]{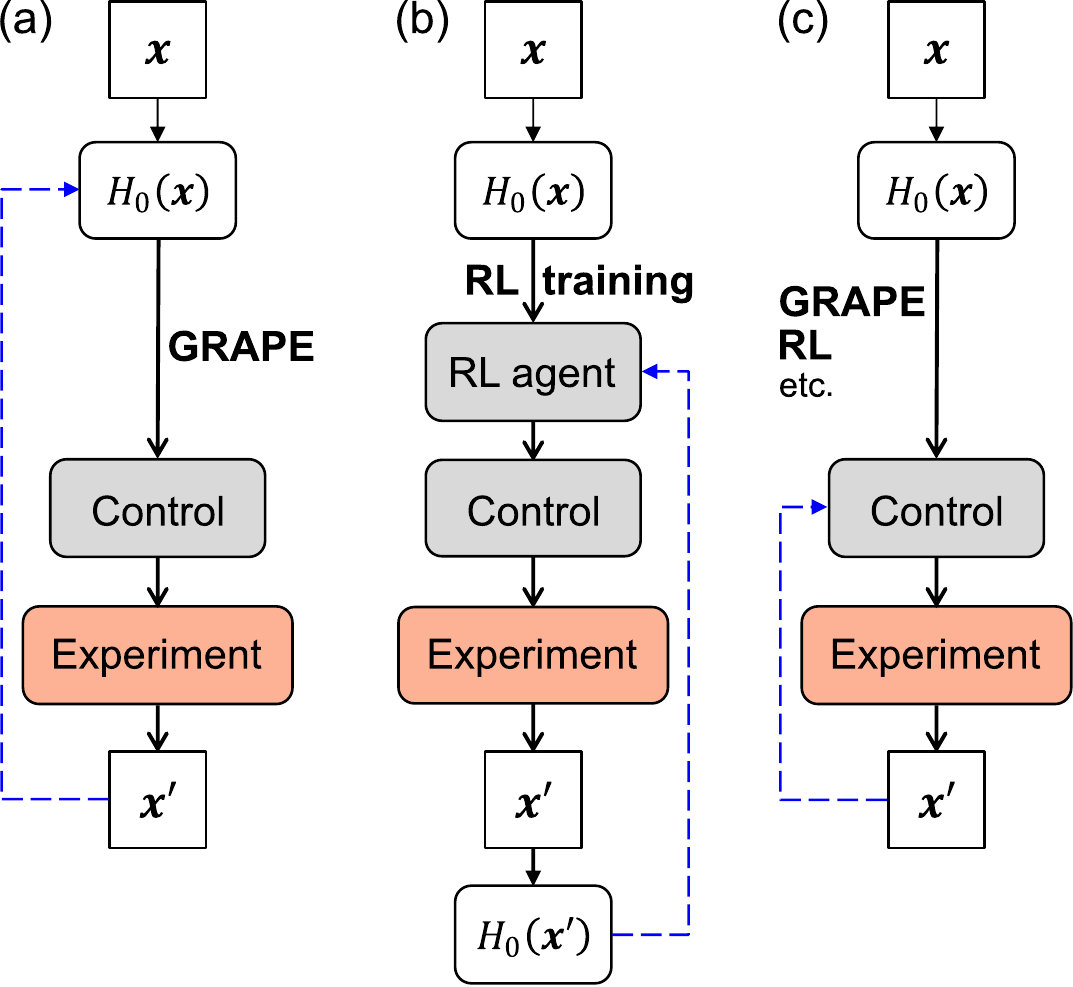}
    \caption{Illustrative flowcharts for estimating $\bm{x}$, a vector encapsulating all parameters to be estimated, using controls optimized by (a) GRAPE, (b) RL, and (c) the analytical scheme. In all panels, the blue dashed line indicates that subsequent control should be optimized according to updated estimations of \black{the} parameters of the previous step. In (a), GRAPE has to be rerun whenever an updated estimation of the parameters is encountered. In (b), the training of the RL agent is only conducted \black{once and} controls for updated parameters can be generated directly from the agent without additional training. In (c), while the initial optimal control has to be found by either GRAPE or RL, subsequent updates can be efficiently done according to Theorem \ref{alg:thm1}.}
    \label{fig:fig1}
\end{figure}

The generalizability afforded by RL \cite{sutton2018reinforcement} has been actively studied in the literature \cite{Tobin2017domain,lee2019network}. While a complete understanding of its origin is still lacking, the present consensus is that the generalizability likely arises from the fact that an RL agent \black{underfits} the training data \cite{sutton1996generalization,Cobbe2018Quantifying,sutton2018reinforcement}. Recently, several schemes have been designed to quantify the generalization in RL \cite{Cobbe2018Quantifying,zhang2018dissection,packer2019assessing}, and it has been found that techniques such as \black{domain} randomization \cite{Tobin2017domain}, \black{network} randomization \cite{lee2019network} and the prediction model \cite{Pathak2017} could improve the generalizability. In quantum parameter estimation, the updated estimations of parameters correspond to the updated environments in RL. \black{Here the} generalization of RL is essentially a problem of training an RL agent that works for an ensemble of systems with ranges of unknown parameters, as we have done in this work.

The main goal of this paper is to elucidate how the generalizability is maintained in the situation of \black{multiparameter} quantum estimation. To this purpose, we introduce an analytical method which adjusts the controls (already optimal for certain parameters) to reverse the shift of the Hamiltonian stemming from the change of the parameters so as to maintain the optimality \black{[see Fig.~\ref{fig:fig1}(c)]}. This analytical method only applies to the situation in which the controls are \black{``complete,''} \black{i.e.,}~there are sufficient control channels to absorb the shift in the Hamiltonian. In this situation,  generalization using the analytical method is most effective as only the initial control for a chosen set of parameters needs to be found by RL or GRAPE. When the shift in the Hamiltonian due to \black{a} change in parameters can only be partially absorbed into available control channels, the analytical method is no longer valid. However, RL still provides generalizable optimal control, albeit for a narrower range of parameters, and we believe that the mechanism for the generalizability afforded by RL is similar to the analytical scheme. \black{Finally}, if the shift in the Hamiltonian cannot be decomposed into available control channels, such as the coupling between two spins considered in this paper, neither RL nor the analytical method provides any generalizability. In this case, one has to rerun either RL or GRAPE for each new set of parameters. Our results provide insights \black{into} when and how the optimal control in \black{multiparameter} quantum metrology can be generalized, thereby facilitating efficient implementation of optimal quantum estimation of multiple parameters, particularly for an ensemble of systems with ranges of parameters.

The remainder of the paper is organized as follows. In Sec.~\ref{sec:model} we briefly review the theory on \black{multiparameter} estimation and present the two examples considered in this paper.  \black{Section}~\ref{sec:algo} explains the methods used in this work, including GRAPE, RL, and the analytical scheme. In Sec.~\ref{sec:gene} we show detailed results on the generalizability applied to the two examples. We \black{summarize} in Sec.~\ref{sec:concl}. Necessary details on the implementation of RL, the comparison between RL and GRAPE, and additional \black{discussion} about the generalizability are provided in \black{the Appendixes}.

\section{multiparameter quantum estimation}\label{sec:model}
We consider the time evolution of a quantum state described by the master \black{equation} \cite{breuer2002theory}
\begin{equation}\label{eq:master}
    \partial_t\rho(t)=\mathcal{L}[\rho(t)],
\end{equation}
where $\mathcal{L}[\rho(t)]=-i[H,\rho(t)]+\black{\Gamma[\rho(t)]}$. $\black{\Gamma[\rho(t)]}$ describes the effect of \black{noises and} $H$ is the Hamiltonian that can be decomposed as
\begin{equation}
    H=H_0(\bm{x})+\sum_{i=1}^P u_i(t)H_i,
\end{equation}
where $H_0(\bm{x})$ is the free evolution with $\bm{x}$, a vector encapsulating the unknown parameters of interest. The \black{term} $u_i(t)H_i$ \black{represents} the control \cite{Khaneja2005Optimal} \black{and} $P$ is the number of available control channels. We discretize the total evolution time $T$ into $N=T/\Delta t$ equal time slices, and a \black{piecewise} constant pulse sequence generates an evolution as
\begin{equation}
    \rho(T)=\prod_{j=0}^{N-1} e^{\Delta t\mathcal{L}_j}\rho(0),
    \label{eq:hcstep}
\end{equation}
where $\mathcal{L}_j$ is the superoperator for the $j$th time slice. The control field is described by the sequence of $u_i(j\Delta t)$ on the interval $[0,T]$.

Given a \black{positive-operator-valued} measurement (POVM) $\qty{\Pi_y}$,
the covariance matrix of the measurements is \black{lower bounded} by \cite{liu2019quantum}
\begin{equation}
    \cov(\widehat{\bm{x}},\Pi_y) \geqslant \frac{1}{n}\mathcal{F}_{\mathrm{cl}}^{-1}(\Pi_y) \geqslant \frac{1}{n}\mathcal{F}^{-1},
\end{equation}
where $\widehat{\bm{x}}$ is the unbiased estimator of $\bm{x}$ and $n$ is the number \black{of repetitions}. The first inequality is the Cram\'er-Rao bound, where $\cfim(\Pi_y)$ is the classical Fisher information matrix which depends on the choice of the measurement,
\begin{equation}
    \mathcal{F}_{\mathrm{cl},\alpha\beta} = \sum_y \frac{\partial_{\bm{x}_{\alpha}}p_{y|\bm{x}}\partial_{\bm{x}_{\beta}}p_{y|\bm{x}}}{p_{y|\bm{x}}},
\end{equation}
where $p_{y|\bm{x}}=\tr[\rho(t)\Pi_y]$. The second inequality is the quantum Cram{\'e}r-Rao bound with \black{$\mathcal{F}$ the} quantum Fisher information matrix, which typically \black{cannot} be attained for the estimation of multiple parameters \cite{liu2019quantum}. In this \black{work we} focus on the classical Cram\'er-Rao (CR) \black{bound and} take the trace of the covariance matrix as the figure of \black{merit}
\begin{equation}\label{eq:CR}
    |\delta\widehat{\bm{x}}|^2 = \tr\cov(\widehat{\bm{x}}) \geqslant \trcfiminv,
\end{equation}
where $\Pi_y$ is omitted from the covariance matrix since a fixed POVM is used for each \black{example and} $\trcfiminv$ is the key quantity governing the precision of measurements that will be frequently referred to simply as \black{the CR bound} in the remainder of the paper.

In the following, we discuss the quantum parameter estimation under local controls on each qubit in the multiple-qubit system. The control Hamiltonian for $m$ qubits is
\begin{equation}\label{eq:localc}
    H_c(t) = \sum_{i=1}^{m} \bm{u}^{(i)}(t)\cdot\bm{\sigma}^{(i)},
\end{equation}
where ${\bm\sigma}^{(i)}=(\sigma_{x}^{(i)},\sigma_{y}^{(i)},\sigma_{z}^{(i)})$ and $\sigma_{a}^{(i)} \equiv \mathbb{I}^{(1)}\otimes \cdots \black{\otimes}\mathbb{I}^{(i-1)}\otimes \sigma_{a}\otimes \mathbb{I}^{(i+1)}\otimes \cdots \otimes\mathbb{I}^{(m)}$. $\mathbb{I}$ is the $2\times2$ identity matrix and $a=x,y,z$. The total number of control channels is $P=3m$.
To compare with the optimal controls obtained with GRAPE, we will use the same examples studied in the previous work \cite{liu2017control}.  


\emph{Example 1.} The first example involves a two-qubit system where one qubit is in a magnetic field with both the magnitude and direction unknown. The other qubit acts as an ancillary qubit. The Hamiltonian of the first qubit is
\begin{equation}\label{eq:h0eg1}
    H_0=\bm B \cdot \bm\sigma^{(1)},
\end{equation}
where $\bm B = (B\sin\vartheta\cos\varphi, B\sin\vartheta\sin\varphi, B\cos\vartheta)$. The parameters to be estimated are $\bm{x}=(B,\vartheta,\varphi)$ characterizing the strength and direction of the magnetic field. The dephasing on the first qubit can be described by the Lindblad form,
\begin{equation}\label{eq:noiseeg1}
    \black{\Gamma[\rho(t)]} = \frac{\gamma}{2}[\sigma_z^{(1)}\rho(t)\sigma_z^{(1)} - \rho(t)],
\end{equation}
where $\gamma$ is the dephasing rate\black{,} which is set as $\gamma=0.2$ in our simulation. Different values of the dephasing rates are considered in Appendix~\ref{sec:appx-rate} for comparison. We use the maximally entangled \black{state $(\ket{00}+\ket{11})/\sqrt{2}$ as} the probe state, and choose the projective measurement on the Bell bases $(\ket{00}\pm\ket{11})/\sqrt{2}$ and $(\ket{01}\pm\ket{10})/\sqrt{2}$.

\emph{Example 2.} The second example is a \black{$ZZ$}-coupled two-qubit system. The Hamiltonian of the free evolution is
\begin{equation}\label{eq:h0eg2}
    H_0=\omega_1\sigma^{(1)}_z+\omega_2\sigma^{(2)}_z+g\sigma^{(1)}_z\sigma^{(2)}_z,
\end{equation}
where $\bm{x}=(\omega_1,\omega_2,g)$ are the parameters to be estimated. \black{Here} $\omega_1$ and $\omega_2$ are individual magnetic fields along the \black{$z$ direction}, while $g$ is the \black{$ZZ$-coupling} strength. Both qubits suffer from dephasing \black{dynamics}
\begin{equation}\label{eq:noise2}
    \Gamma[\rho(t)]=\sum_{i=1,2}\frac{\gamma_i}{2}[\sigma^{(i)}_z\rho(t)\sigma^{(i)}_z -\rho(t)],
\end{equation}
where we set the dephasing rates $\gamma_1=\gamma_2=0.1$ in this work. Other values of the dephasing rates are studied in Appendix~\ref{sec:appx-rate}.
Under unitary dynamics (\black{i.e.,}~without noise), the optimal probe state for this system is $\ket{++}=(\ket{00}+\ket{01}+\ket{10}+\ket{11})/2$. We thus use $\ket{++}$ as the probe and  the local measurements $\dyad{++}{++}, \dyad{+-}{+-}, \dyad{-+}{-+}$\black{,} and $\dyad{--}{--}$\black{,} where $\ket{\pm}=(\ket{0}\pm\ket{1})/\sqrt{2}$.

To facilitate the discussion, all parameters in this paper are considered dimensionless unless \black{indicated otherwise}, \black{e.g.,} $B$, $\omega_1$, $\omega_2$\black{, and $g$. The} control fields are measured in the energy unit of \black{1 and} the time $t$ is evaluated in the time unit of 1.

\vspace{16pt}

\section{Methods}\label{sec:algo}
\subsection{Gradient ascent pulse engineering}\label{subsec:grape}
For the case with only two parameters,  Eq.~\eqref{eq:CR} reduces to $\abs{\delta\widehat{\bm{x}}}^2 \geqslant \det\cfim/\tr\cfim$. It is possible to derive the gradient of the CR bound with respect to the control field and perform GRAPE to find appropriate optimal controls \cite{liu2017control}. However, for the cases with more than two parameters, \black{like the two examples above}, taking the gradient of $\trcfiminv$ is not straightforward. \black{Reference}~\cite{liu2017control} \black{introduced} another objective \black{function $f_0(T)$ as}
\begin{equation}
    f_0(T)^{-1}=\sum_{\alpha}\frac{1}{\mathcal{F}_{\mathrm{cl},\alpha\alpha}(T)},\label{eq:ofunc}
\end{equation}
which provides a lower bound to $\trcfiminv$ since $\qty[\cfiminv]_{\alpha\alpha} \geqslant 1/\mathcal{F}_{\mathrm{cl},\alpha\alpha}$ and $\trcfiminv \geqslant \sum_{\alpha}1/\mathcal{F}_{\mathrm{cl},\alpha\alpha}$. Practically, the optimal controls are found by optimizing $f_0(T)$, even though it is not necessarily the optimization of the desired CR \black{bound $\trcfiminv$}.
While we largely follows Ref. \cite{liu2017control} to repeat the results obtained with GRAPE, we additionally apply the steepest gradient ascent method \cite{ruder2016overview}  and the Adam method \cite{kingma2014adam} to update the control field, which are numerically more efficient than the simple updating scheme employed in Ref. \cite{liu2017control}.

The evolution of the density matrix must be evaluated $N^2$ times to compute the gradient with respect to control, while there are $N$ piecewise controls that need to be updated \cite{xu2019generalizable}. Though the time complexity of GRAPE can be reduced from \black{$O{(N^3)}$ to $O{(N^2)}$} by spending more computer memory on storing the superoperators, the GRAPE method, as can be seen in Sec.~\ref{sec:gene}, quickly becomes \black{prohibitively} expensive as the size of the problem increases.

\subsection{Reinforcement learning}

\begin{figure}
    \includegraphics[width=0.8\linewidth]{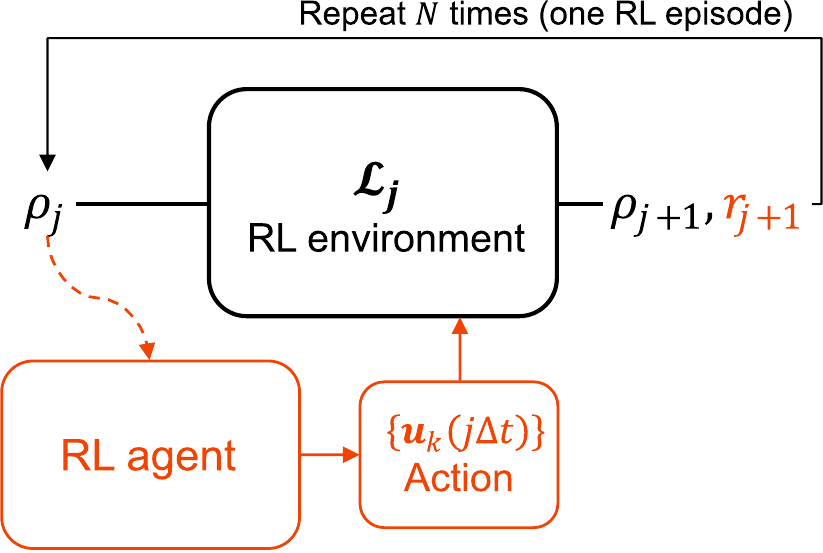}
    \caption{Schematics of the reinforcement learning framework. At the $j$th time step, the RL agent observes the current state as the density matrix $\rho_j\equiv\rho(j\Delta t)$ and takes an action as the control on the system. As a consequence, the state evolves to a new one at the $(j+1)$th time step while the agent receives an immediate reward. After completing one episode, the state is reset to the probe state.}
    \label{fig:fig2}
\end{figure}

For the \black{piecewise} control in Eq.~\eqref{eq:hcstep}, we can interpret the optimization problem as a Markov decision process in RL. \black{Figure~\ref{fig:fig2}} shows the schematics of RL.
The RL procedure is essentially an interaction between an agent and its environment, the latter of which involves a state space $\mathcal{S}$, an action space $\mathcal{A}$, and a numerical reward $\mathcal{R}$ \cite{sutton2018reinforcement}.
The behavior of the agent is defined by a \textit{policy}. In this \black{work we} use the deterministic policy, \black{namely,} it always maps $s_j\in\mathcal{S}$ to a specific action $a_j\in\mathcal{A}$ \cite{silver2014deterministic}.
At each time step $j$, the agent observes the current state $s_j$ and chooses an action $a_j\in \mathcal{A}$ according to its policy. The action $a_j$ constitutes the control $\bm u(j\Delta t)$ which steers the quantum state according to the dynamics of the environment described by the \black{time evolution~\eqref{eq:master}}, while the agent receives an immediate reward $r_{j+1}$. The agent then observes \black{the} state $s_{j+1}$ and conducts action $a_{j+1}$ prescribed by the environment. This procedure is iterated until the target time $T$ (or total number of \black{steps} $N$) is reached, concluding one \textit{episode} of RL.  Each episode is summarized by a \textit{trajectory} of states, actions, and rewards as $\qty(s_0,a_0,r_1,s_1,\dots,a_{N-1},r_{N},s_{N})$. A replay memory is used to store $(s_j,a_j,r_{j+1},s_{j+1})$ (\black{refered to as a sample}) at each time step in the iteration. The RL agent is then trained using samples from the replay memory, which stabilizes the learning process \cite{lin1993reinforcement,mnih2013playing}.

Since the action $a_j$ influences not only the immediate reward $r_{j+1}$ but also the future rewards $r_{k>j+1}$, the immediate reward alone is not an appropriate objective for training the RL agent. To take into account the future rewards, we define the total discounted reward $R_j=\sum_{k=j}^N\alpha^{k-j}_{\rm rl} r_{k+1}$ with $\alpha_{\rm rl}\in[0,1)$ the discount rate \cite{sutton2018reinforcement}. In practice, noises must be added to the actions chosen so as to prevent the simulation being trapped in a local minimum \cite{silver2014deterministic,lillicrap2015continuous}. Therefore, $R_j$ can have different values across different episodes, depending on the policy and the dynamics of the environment. We therefore use the expectation value $R_j$ across all episodes associated \black{with} the given deterministic policy and dynamics, $\mathbb{E}\qty[R_j]$, as our objective. Practically, the RL agent maximizes the expected value of the total discounted reward at the initial state, $\mathbb{E}\qty[R_0]$ \cite{sutton2018reinforcement}.

In this \black{work we} assume that the density \black{matrix $\rho(j\Delta t)$ of} any quantum state concerned is fully known to the agent \cite{Thomas2018Reinforcement,xu2019generalizable,schuff2020improving}, which may not be true in experiments \cite{leonhardt1997measuring,James2001measurement}.
Our numerical simulation can be viewed as an emulator with \black{regard} to the RL procedure.
In our numerical simulation, we must guess a value of the parameter to build the emulator for the RL agent. 
During one episode (either training or testing) of RL, the master equation is evaluated $N$ times in the RL environment. The time complexity for each episode is therefore \black{$O{(N)}$}. For a given problem, the training may actually be slower than GRAPE as a large prefactor is added depending on the details of the algorithm. Nevertheless, the generalizability afforded by RL can be advantageous in certain situations.

We employ the \black{deep deterministic policy gradient (DDPG) reinforcement learning} algorithm \cite{lillicrap2015continuous}, which uses neural networks as function approximators of the agent. The agent's policy is chosen to be deterministic for the following reasons. \black{First}, it can treat  continuous actions as required in our problems. \black{Second}, computing the deterministic policy requires \black{fewer} samples than otherwise \cite{silver2014deterministic}, which saves computational resources. \black{Third}, the main subject of our study is the generalizability, and having deterministic policies greatly facilitates evaluating the  generalization of RL. The implementation of the DDPG procedure is detailed in Appendix~\ref{sec:appx-ddpg}.

\subsection{Analytical scheme}

We now introduce the method of adjusting the optimal controls obtained for the parameters $\bm{x}$ to the new \black{one $\bm{x}'=\bm{x}+\delta \bm{x}$ so} that it can offer a high precision of estimation at the new parameters. We emphasize that this method can only be used on a known initial control optimized at a specific parameter, with \black{complete} control channels available to reverse changes in the Hamiltonian.
The Hamiltonian of the new value $\bm{x}'$ can be rewritten as
\begin{equation}
    H(\bm{x}') = H_0(\bm{x}) + \delta H{(\bm{x},\bm{x'})} +\sum_{i=1}^P u_i(t)H_i,
\end{equation}
where $\delta H{(\bm{x},\bm{x'})} = H_0(\bm{x}') - H_0(\bm{x})$ is the change of the Hamiltonian of free evolution. \black{Here} $H_c(t)=\sum_{i=1}^P u_i(t)H_i$ is the control Hamiltonian on \black{a} total of $P$ channels.
We are interested in the piecewise square pulses of Eq.~\eqref{eq:hcstep}
on the time interval $[0,T]$. The analytical method that generalizes optimal control to other parameters relies essentially on the following \black{theorem.}

\begin{thm}\label{alg:thm1}
  For any $\bm{x}'$, if $\delta H{(\bm{x},\bm{x'})}$ can be decomposed into the linear combination $\delta H{(\bm{x},\bm{x'})} = \sum_i \delta u_{\bm{x}\bm{x'},i} H_i$ \black{and} the optimal control $\qty{u_i}$ is known for the Hamiltonian $H_0(\bm{x})$, the optimal control for the Hamiltonian $H_0(\bm{x}')$ is
    \begin{equation}\label{eq:opt_control}
        u_i \to u_i - \delta u_{\bm{x}\bm{x'},i}.
    \end{equation}
\end{thm}
\begin{proof}
    When the control field $\qty{u_i}$ changes to $u_i-\delta u_{\bm{x}\bm{x'},i}$ on each time slice,
    \begin{eqnarray}
            H(\bm{x}')&=&H_0(\bm{x}) + \delta H{(\bm{x},\bm{x'})} +\sum_{i=1}^P \qty(u_i-\delta u_{\bm{x}\bm{x'},i})H_i \nonumber\\
            &=& H_0(\bm{x})  +\sum_{i=1}^P u_iH_i \equiv H(\bm{x}),\label{eq:hprime}
    \end{eqnarray}
    and the superoperators describing the dynamics with respect to $\bm{x}'$ are $\mathcal{L}_j(\bm{x}')=\mathcal{L}_j(\bm{x})$. Consequently, the states at each time step are $\rho_{\bm{x}'}(j\Delta t)=\rho_{\bm{x}}(j\Delta t)$ according to the state evolution defined by Eq.~\eqref{eq:hcstep}. We can write down the derivative with respect to the parameters $x'_{\alpha}$ at the final state in the \black{first-order approximation}
    \begin{equation}\label{eq:drho}
        \pdv{\rho_{\bm{x}'}(T)}{x'_\alpha} = \sum_{j=0}^{N-1} D_{j}^{N-1} \pdv{\Delta t \mathcal{L}_j(\bm{x}')}{x'_\alpha}  \rho_{\bm{x}}(j\Delta t),
    \end{equation}
    where $D_{j}^{N-1}=e^{\Delta t  \mathcal{L}_{N-1}}\cdots e^{\Delta t  \mathcal{L}_{j+1}}e^{\Delta t  \mathcal{L}_j}$.
    The partial derivative in \black{Eq.~\eqref{eq:drho}} is simply the derivative over the same superoperator but with respect to the new parameters $\bm{x}'$. First, we define a transformation, $\partial_{x'_\alpha}\mathcal{L}_j(\bm{x}')=\partial_{x_\beta}\mathcal{L}_j(\bm{x}){R}_{\beta\alpha}$ with summation over repeated indices, that maps the two partial derivatives to each other as follows.


    When the partial derivative has an implicit form, we perform a change of variables
      $\partial_{x_\alpha}\mathcal{L}_j(\bm{x}) = \partial_{z_\beta}\mathcal{L}_j(\bm{x}){C}_{\beta\alpha}(\bm{x})$,
    where the Jacobian matrix ${C}_{\beta\alpha}(\bm{x}) = \partial z_\beta(\bm{x})/ \partial x_\alpha$ and its inverse ${C}_{\alpha\beta}^{-1}(\bm{x}) = \partial x_\alpha(\bm{x})/ \partial z_\beta$. Although the optimal control and $\delta u_{\bm{x}\bm{x'},i}$ naturally depend on the parameters $\bm{x}$ and $\bm{x'}$, we stress that the control field $u_i$ in Eq.~\eqref{eq:hprime} is independent of $\bm{x}$ and $\bm{x'}$, because $u_i$ has been optimized for the particular $\bm{x}$ and is kept constant, so $\partial_{z_\beta}\mathcal{L}_j(\bm{x})=-i\partial_{z_\beta}H_0^{\cross}(\bm{x})$ with $H_0^{\cross}(\bm{x})\ \bullet \equiv\qty[H_0(\bm{x}),\bullet]$. Such new variables should be used such that $\partial_{z_\beta}H_0^{\cross}$ has an explicit form and satisfies $\partial_{z_\beta}H_0^{\cross}(\bm{x}')=\partial_{z_\beta}H_0^{\cross}(\bm{x})$.
    We can then rewrite $\partial_{x'_\alpha}\mathcal{L}_j(\bm{x}')$ and obtain the transformation matrix ${R}_{\beta\alpha}$,
    \begin{equation}\label{eq:transR}
      \begin{aligned}
        \partial_{x'_\alpha}\mathcal{L}_j(\bm{x}') &= \partial_{x_\beta}\mathcal{L}_j(\bm{x}){C}^{-1}_{\beta\zeta}(\bm{x}){C}_{\zeta\alpha}(\bm{x}') \\
        &\equiv \partial_{x_\beta}\mathcal{L}_j(\bm{x}){R}_{\beta\alpha}.
      \end{aligned}
    \end{equation}
    Substituting \black{Eq.~\eqref{eq:transR}} in Eq.~\eqref{eq:drho}, we have
    \begin{equation}
        \pdv{\rho_{\bm{x}'}(T)}{x'_\alpha} = \sum_{\beta}\sum_{j=0}^{N-1} D_{j}^{N-1} \pdv{\Delta t \mathcal{L}_j(\bm{x})}{x_\beta}  \rho_{\bm{x}}(j\Delta t){R}_{\beta\alpha}.
    \end{equation}
    Therefore, the change in Eq.~\eqref{eq:opt_control} leads to $\partial_{x'_\alpha}\rho_{\bm{x}'}(T) = {R}_{\beta\alpha} \partial_{x_\beta}\rho_{\bm{x}}(T)$.
    
    \black{Next we} calculate the classical CR bound for the parameters $\bm{x}$ and $\bm{x}'$ using the POVM $\qty{\Pi_y}$.  Using the relation $\rho_{\bm{x}'}(j\Delta t)=\rho_{\bm{x}}(j\Delta t)$ and the definition of the measurement $p_{y\vert\bm{x}}=\tr[\rho_{\bm{x}}(T)\Pi_y]$, we obtain the measurement $p_{y\vert\bm{x}'}=p_{y\vert\bm{x}}$. Since $\partial_{x'_\alpha}\rho_{\bm{x}'} = {R}_{\beta\alpha}\partial_{x_\beta}\rho_{\bm{x}}$,
    $\partial_{x_\alpha}p_{y\vert\bm{x}}=\tr[\partial_{x_\alpha}\rho(T)\Pi_y]$ yields $\partial_{x'_\alpha}p_{y\vert\bm{x}'} = {R}_{\beta\alpha}\partial_{x_\beta}p_{y\vert\bm{x}}$. Thus, the classical Fisher information matrices satisfy $\cfim(\bm{x}') = {R}^{T}\cfim(\bm{x}){R}$ or $\cfiminv (\bm{x}')={R}^{-1}\cfiminv (\bm{x})\qty({R}^{-1})^{T}$.
    Taking the trace of both sides leads to
    \begin{equation}\label{eq:lemma}
      \trcfiminv(\bm{x}') = \tr[{R}^{-1}\cfiminv (\bm{x})\qty({R}^{-1})^{T}].
    \end{equation}
    Regarding the parameters $\bm{x}$, the control $\qty{u_i}_{\rm opt}$ is optimal in the sense that $\trcfiminv(\bm{x}\vert \qty{u_i}_{\rm opt}) \leqslant \trcfiminv(\bm{x}\vert \qty{u_i})$. Since $\cfiminv$ is positive definite,  Eq.~\eqref{eq:lemma} guarantees $\trcfiminv(\bm{x'}\vert \qty{u_i}_{\rm opt}) \leqslant \trcfiminv(\bm{x'}\vert \qty{u_i})$, \black{i.e.,} the adjusted control of Eq.~\eqref{eq:opt_control} is optimal for the parameters $\bm{x}'$.

    Note that the theorem relies on the fact that the available control channels must be capable \black{of absorbing} the changes in the Hamiltonian, requiring \black{complete} channels of control. When the control is limited, \black{e.g.,} only along one particular direction, the theorem is no longer applicable.

\end{proof}

\section{Generalizability}\label{sec:gene}

In this \black{section we compare} results from different methods for ensembles of systems where the parameters take a range of values. In this case, the controls and estimated values of parameters must be updated adaptively \cite{liu2017quantum}.
For the analytical method, we firstly optimize the controls for a chosen set of parameters using RL and GRAPE, and then analytically shift the controls with the new set of \black{parameters} according to Theorem \ref{alg:thm1}.
Since the analytical scheme directly modifies the control profile, it provides the most efficient method to generalize an optimal control to other parameters. For RL, an emulator for the new estimated value has to be created. The agent runs in the emulator and designs the quantum control by taking the states of each time step, which has a small computational cost \black{$O{(N)}$}.
We \black{will} see that the generalizability of RL is retained for most parameters to be measured. We also consider the situation where the control channels are incomplete so the analytical scheme becomes inapplicable, while the generalizability of the RL agent \black{is retained} in a smaller range. On the other hand, for the change of the \black{$ZZ$-coupling} strength in Example 2, both the analytical scheme and the generalizability of RL fail, so the controls must be optimized by RL or GRAPE individually for each new value.

\begin{figure}[tb]
    \includegraphics[width=0.8\linewidth]{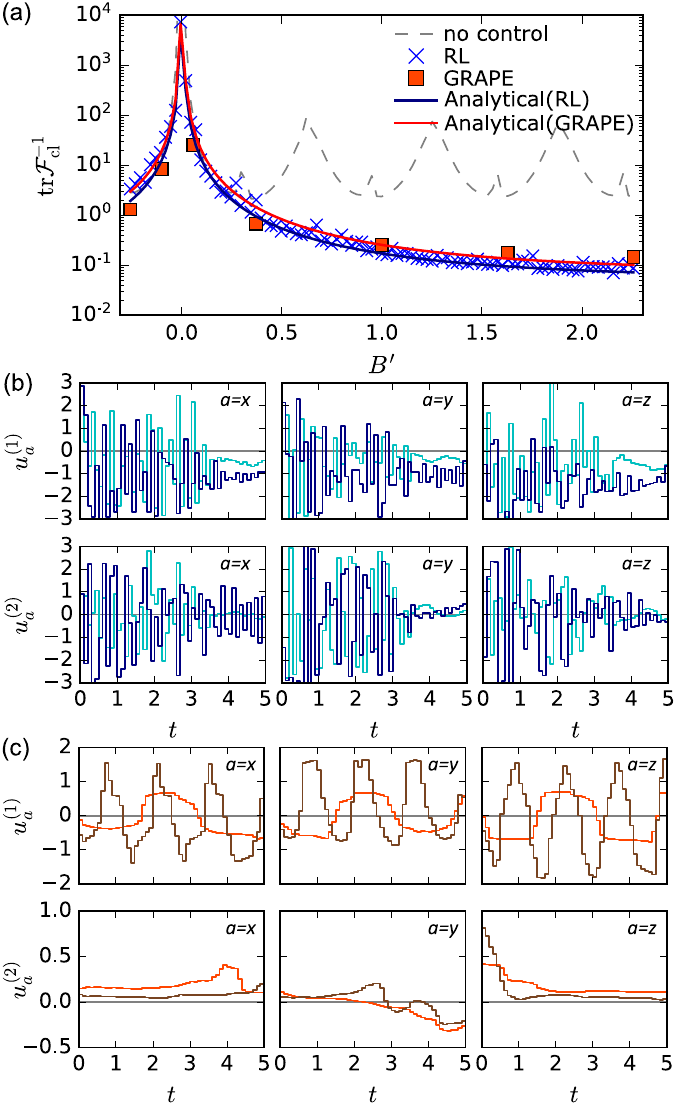}
    \caption{(a) Generalizability for the strength of the magnetic field in Example 1. For RL, we show $\trcfiminv$ at the target time $T=5$ as a function of $B'$ in $[B-2\pi/T,B+2\pi/T]$ where $B=1$. We use the neural networks that are trained at $(B,\vartheta,\varphi)=(1,\pi/4,\pi/4)$ to generate the pulse sequences for $B'$. The blue crosses \black{show} the results of RL. The red squares represent the results using new pulse sequences optimized by GRAPE. The solid dark blue (red) curve show the analytical results using the pulse sequences found by RL (GRAPE) at  $(B,\vartheta,\varphi)=(1,\pi/4,\pi/4)$. The dashed gray curve shows the CR bound without control. (b) Representative pulse profiles obtained with RL. The cyan lines show the results of $B'=B$. The purple lines show the results of $B'=B+2\pi/T$. (c) Representative pulse profiles obtained with GRAPE. The red lines show the results of $B'=B$. The brown lines show the results of $B'=B+2\pi/T$.}
    \label{fig:fig3}
\end{figure}

\textit{Example 1.}
For the analytical method, \black{due} to Theorem~\ref{alg:thm1}, once we have the optimal control $\bm{u}^{(1)}$ for $\bm{B}$, we can update the control field as $\bm{u}^{(1)} \to \bm{u}^{(1)}-\qty(\bm{B'}-\bm{B})$ when the estimation of the parameter is updated to $\bm{B'}$.
At the same time, the RL policy is learned specifically at $\bm{B}$\black{; then} we generalize the trained policy network to the different magnetic field $\bm{B'}$.
In the simulation, the control fields and the RL agents are optimized at $(B,\vartheta,\varphi)=(1,\pi/4,\pi/4)$.

One of the simplest cases is that $\bm{B}$ and $\bm{B'}$ are along the same direction. In \black{Theorem~\ref{alg:thm1} the} matrix ${R}$ consists of the ratios ${R}_{BB}=1$, \black{${R}_{\vartheta\vartheta}=B'/B$, and} ${R}_{\varphi\varphi}=B'/B$. All other elements are zero. The CR bound at the new value of the parameter is given by
\begin{equation}\label{eq:gene_B}
    \trcfiminv(\bm{B}')=\delta \widehat{B}^2+\qty(\frac{B}{B'})^2\qty(\delta \widehat{\vartheta}^2+\delta \widehat{\varphi}^2),
\end{equation}
where $\delta \widehat{B}^2$, \black{$\delta \widehat{\vartheta}^2$, and} $\delta \widehat{\varphi}^2$ represent the diagonal terms of $\cfiminv(\bm{B})$. \black{Here} $\trcfiminv$ diverges at $B'=0$, but with the increase of $B'$ the value of $\trcfiminv$ decreases and it is lower bounded by $\delta \widehat{B}^2$.

\black{Figure}~\ref{fig:fig3}(a) shows the CR bound at the target time $T$ as a function of $B'$, where $B'$ is allowed to vary \black{within} $[B-2\pi/T,B+2\pi/T]$, while the direction of the magnetic field is fixed at $(\vartheta,\varphi)=(\pi/4,\pi/4)$. In \black{Fig.~\ref{fig:fig3}(a) the} dark blue and red solid lines correspond to the analytical scheme using the optimized controls obtained with RL and GRAPE\black{, respectively, at $B=1$} (for \black{the} choice of the target time $T$ along with other details in the implementation of RL, see Appendix~\ref{sec:irl}).
\black{The results of} evaluating the generalization of RL at different $B'$ are shown by blue crosses. In addition, the CR bound obtained by rerunning the GRAPE procedure at selected $B'$ is presented \black{as} red squares.
\black{The results of} the CR bound from the generalization of RL \black{are} substantially lower than \black{for} the case without control in a large neighborhood around $B'=1$, confirming the generalizability. In fact, for $B'>1$, RL outperforms GRAPE, even if it is not trained \black{or} optimized specifically at those points. Moreover, the analytical scheme and the generalization of RL return similar values of the CR bounds, indicating that controls produced by a trained RL agent away from the training point is near optimal. For $|B'|\lesssim0.3$, no method offers notable improvement as the optimal solution diverges according to Eq.~\eqref{eq:gene_B}.
We note that each of the GRAPE \black{points} costs on average \black{more than or approximately $2$ h}, meaning that rerunning the optimization in a range of $B'$ would be prohibitively expensive. This demonstrates the effectiveness of RL and the analytical approaches.

\black{Figure}~\ref{fig:fig3}(b) shows representative pulse profiles found by RL, comparing the results found at the training point $B'=B=1$ (cyan lines) and at a point away from the training point $B'=B+2\pi/T$ (purple lines). Similarity between the two sets of pulses $\bm{u}^{(1)}$ is evident from the figure, suggesting that  RL achieves generalizability through a mechanism similar to Theorem~\ref{alg:thm1}, \black{i.e.,} the reversal of changes of the free evolution by adjusting available controls.

Fixing the strength $B$, we consider the change in the direction $(\vartheta,\varphi)$. \black{Here we} apply the change of variables transforming the spherical coordinates to Cartesian \black{notation} \cite{yuan2016sequential}. The Jacobian matrix and ${R}$ are much more complicated, and the classical Fisher information matrix is not necessarily diagonal for $\bm{B'}$.
Nevertheless, for $B'=B=1$ and varying $(\vartheta',\varphi')$, the CR bound of $\bm{B'}$ is given by
\begin{equation}
    \trcfiminv(\bm{B}') = C_1\delta \widehat{B}^2 + C_2\delta \widehat{\vartheta}^2 + C_3\delta \widehat{\varphi}^2,
 \end{equation}
  with
\begin{align*} \begin{cases}
     C_1=\cos^2\vartheta+\sin^2\vartheta\qty(\cos^2\Delta\varphi+\csc^2\vartheta'\sin^2\Delta\varphi), \\
     C_2=\sin^2\vartheta+\cos^2\vartheta\qty(\cos^2\Delta\varphi+\csc^2\vartheta'\sin^2\Delta\varphi ), \\
     C_3=\sin^2\vartheta\csc^2\vartheta' \cos^2\Delta\varphi+\sin^2\vartheta\sin^2\Delta\varphi,
    \end{cases}  \end{align*}
where $\Delta\varphi=\varphi'-\varphi$ is the change in $\varphi$. This expression shows that $\Delta\varphi$ hardly shifts the optimal value of the CR bound, but $\vartheta'$ has notable influences on the results. For example, if $\varphi'=\varphi$ and $\vartheta'$ varies, the CR bound would be
\begin{equation}\label{eq:gene_theta}
    \trcfiminv(\bm{B}') = \delta \widehat{B}^2 + \delta \widehat{\vartheta}^2 + \black{\sin^2{\vartheta}\csc^2{\vartheta'}}\,\delta \widehat{\varphi}^2,
\end{equation}
and $\cfiminv(\bm{B}')$ is \black{block diagonal}. The change of \black{the} parameter $\vartheta'$ can have \black{a} strong influence on the estimation of $\varphi'$. Specifically, the new lower bound $\delta\widehat{\varphi}'^2$ has the smallest value at $\vartheta'=\pi/2$ and diverges at $\vartheta'=0,\pi$, \black{i.e.,} the poles $\ket{0}$ \black{and} $\ket{1}$ on the Bloch sphere.

In Fig.~\ref{fig:fig4} we show the results from the analytical scheme and the generalization of RL concerning the direction of the magnetic field, \black{i.e.,}~$\vartheta'$ and $\varphi'$ are varied within $[0,\pi]$ and $[0,2\pi]$, respectively. The strength of the magnetic field is fixed at $B=1$ and RL is trained at $(\vartheta,\varphi)=(\pi/4,\pi/4)$.
\black{Figure}~\ref{fig:fig4}(a) shows the \black{pseudo-color-plot} of the CR bound. In a reasonably large neighborhood around the original training point
$(\vartheta,\varphi)=(\pi/4,\pi/4)$, the RL-designed pulse sequences produce \black{values of $\trcfiminv$ similar} to \black{those of} the analytical method, indicating generalizability.
\black{Figure}~\ref{fig:fig4}(b) and (c) each \black{present} a cut at fixed $\vartheta'$ or $\varphi'$ as indicated by the dashed lines in the right panel of Fig.~\ref{fig:fig4}(a). \black{Figure}~\ref{fig:fig4}(b) shows results with $\vartheta'=\pi/4$ and $\varphi'$ \black{varying in} $[0,2\pi]$. For the analytical method, it is clear that results initiated by RL have lower CR bounds as compared to the case with no control and GRAPE. This is also true in Fig.~\ref{fig:fig4}(c), which shows results with $\varphi'=\pi/4$ and varying $\vartheta'$ in $[0,\pi]$. Overall, quantum controls suppress the divergences that appear periodically in the no-control case at $\varphi'=0,\pi$ or $\vartheta'=\pi/2$, although the value of \black{the} CR bounds \black{increases} to infinity as $\vartheta'$ \black{approaches} $0$ or $\pi$, as expected from Eq.~\eqref{eq:gene_theta}. Scattered points showing results from GRAPE are found to be comparable to its corresponding analytical results, but we emphasize that here each point is a complete \black{rerun} of the full algorithm which is very expensive. On the other hand, RL provides \black{suboptimal} results with a much lower cost on the computational resources.

\begin{figure}
  \includegraphics[width=0.8\linewidth]{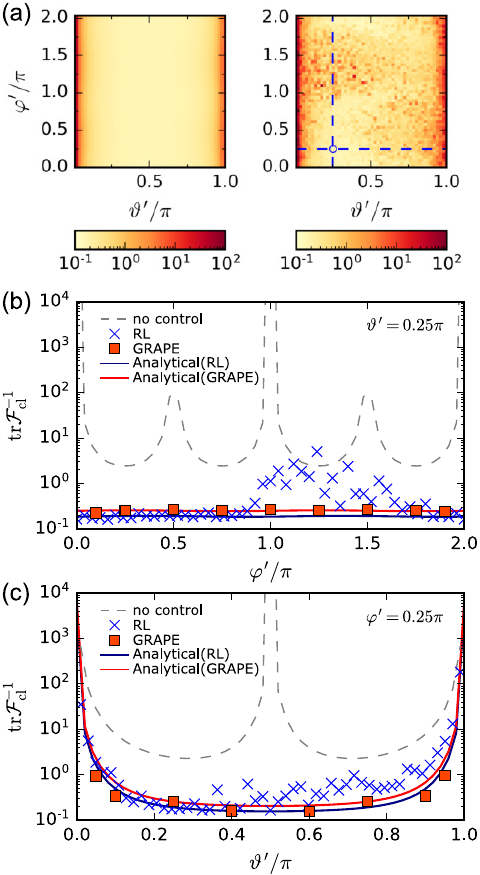}
  \caption{Generalizability for the direction of the magnetic field in Example 1. The strength of the magnetic field is fixed as $B=1$.  (a) \black{Pseudo-color-plots} showing the CR \black{bound $\trcfiminv$ at} the target time $T=5$ as \black{a function} of the parameters $(\vartheta',\varphi')$ in the plane $[0,\pi]\times [0,2\pi]$. Note that the color bar has a \black{logarithmic} scale. \black{Shown on the left is} the CR bound using the analytical scheme. \black{Shown on the right is} the CR bound under controls generated by RL using the policy network trained at $(B,\vartheta,\varphi)=(1,\pi/4,\pi/4)$, the location of which is denoted by a blue circle. Two  cuts (\black{shown} as blue dashed lines) correspond to \black{the} results \black{in} (b) and (c). (b) The CR bound as \black{a function} of $\varphi'$ at fixed $\vartheta'=0.25\pi$. (c) The CR bound as \black{a function} of $\vartheta'$ at fixed $\varphi'=0.25\pi$.}
 \label{fig:fig4}
\end{figure}

\begin{figure}
    \includegraphics[width=0.8\linewidth]{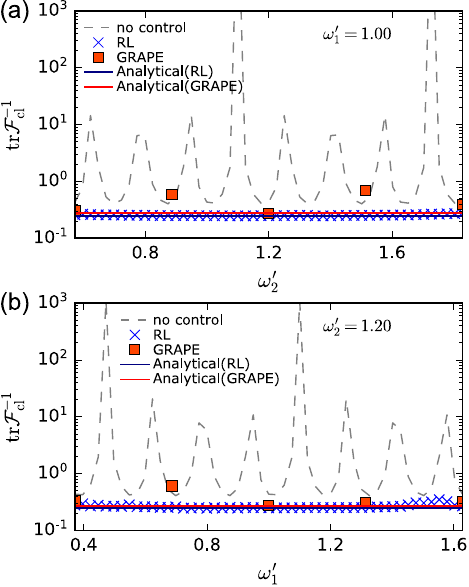}
    \caption{Generalizability for the individual magnetic fields in Example 2. The \black{$ZZ$-coupling} strength is fixed at $g=0.1$.  (a) The CR bound as \black{a function} of $\omega_2'$ at fixed $\omega_1'=1.0$. (b) The CR bound as \black{a function} of $\omega_1'$ at fixed $\omega_2'=1.2$.}
   \label{fig:fig5}
\end{figure}

\textit{Example 2.} The simplest case for coupling between qubits is the \black{$ZZ$ coupling} as in Eq.~\eqref{eq:h0eg2}. Again, we introduce a change to the \black{parameters $\bm{x}\to\bm{x}'$}, so $\delta H{(\bm{x},\bm{x'})}=H_0(\bm{x}')-H_0(\bm{x})= \sum_{k=1}^2\qty(\omega_k'-\omega_k)\sigma_z^{(k)}+\qty(g'-g)\sigma_z^{(1)}\sigma_z^{(2)}$. 
When the individual magnetic fields $\omega_{k=1,2}$ change, it is straightforward to decompose the change $\delta H{(\omega_k,\omega_k')}$ into local controls $\sigma_z^{(k)}$. In the analytical scheme, control fields are adapted as $u_z^{(k)} \to u_z^{(k)}-\qty(\omega_k'-\omega_k)$. Note that the partial derivative of $\omega_k\sigma_z^{(k)}$ with respect to $\omega_k$ is not a function of the parameter itself, \black{i.e.,}~${R}=\mathbb{I}_{3}$ is a $3\times3$ identity matrix. Thus, Theorem~\ref{alg:thm1} proves that the optimal CR bound always stays at the same \black{value $\trcfiminv(\bm{x}')=\trcfiminv(\bm{x})$}. For the generalization of RL, the RL agent is initially trained at the parameters $(\omega_1,\omega_2,g)=(1.0,1.2,0.1)$. Similar to Example~1, the trained RL agent will then be straightforwardly used on a range of $\omega_1'$ and $\omega_2'$.

In Fig.~\ref{fig:fig5}, we fix the \black{$ZZ$-coupling} strength and show the values of the CR bound with varying  individual magnetic fields. To better illustrate the results, we fix one of $\omega_1'$ and $\omega_2'$ and vary the other. In Fig.~\ref{fig:fig5}(a), $\omega_1'$ are fixed at 1.0 while $\omega_2'$ varies, and in Fig.~\ref{fig:fig5}(b) we fix $\omega_2'$ at 1.2 and change $\omega_1'$. In both cases, the CR bounds obtained by the analytical method stay at constant values as noted above. Within the analytical scheme, initial controls obtained by RL give a slightly lower value of the CR bound. Meanwhile, RL shows a high level of generalizability since the CR bound is orders of magnitudes lower than that of no-control cases and is almost the same as those from the analytical method for a wide range of values. We have used GRAPE to optimize new control fields for selective $\omega_k'$ and the results are \black{represented} by red squares. The fluctuation in the optimized values may arise from the fact that the optimization procedure of GRAPE might fall in a local \black{minimum and} that $f_0(T)$ \black{[Eq.~\eqref{eq:ofunc}]} is different from the desired CR bound as noted in Sec.~\ref{subsec:grape}.

\begin{figure}[tb]
    \includegraphics[width=0.8\linewidth]{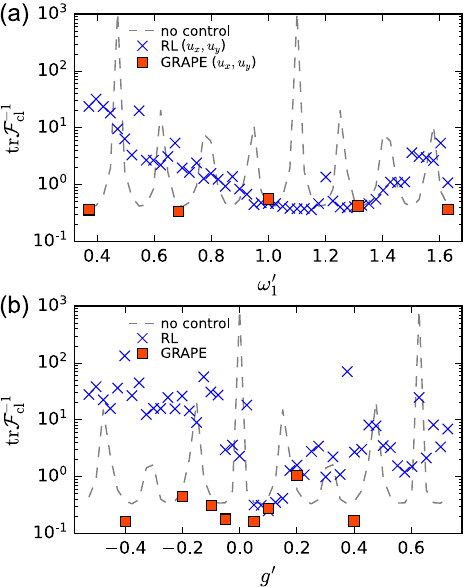}
    \caption{(a) Generalizability for the magnetic field using only the local controls $\sigma_x^{(1)}$, $\sigma_y^{(1)}$, $\sigma_x^{(2)}$ and $\sigma_y^{(2)}$. We show $\trcfiminv$ at the target time $T=5$ as a function of $\omega_1'$ in $[1.0-\pi/T,1.0+\pi/T]$. \black{The other} parameters are fixed at $\omega_2'=1.2$ and $g=0.1$. (b) Lack of generalizability for the \black{$ZZ$-coupling} strength. For RL, we show $\trcfiminv$ at the target time $T=5$ as a function of $g'$ in $[0.1-\pi/T,0.1+\pi/T]$. \black{The magnetic} fields $\omega_1'=1.0$ and $\omega_2'=1.2$.}
    \label{fig:fig6}
\end{figure}

Under Example 2, we discuss two representative situations where the analytical method becomes inapplicable, where one must either apply GRAPE \black{for} each set of parameters or utilize the generalizability of RL.

In the first situation, we restrict the control channels to $\sigma_x^{(1)}$, $\sigma_x^{(2)}$, $\sigma_y^{(1)}$ and $\sigma_y^{(2)}$, so the shift $\delta H{(\omega_k,\omega_k')}\sim\sigma_z^{(k)}$ \black{resulting from} the change of $\omega_k$ can no longer be absorbed \black{by} local controls. In order to test the generalizability, the RL agent is trained at $(\omega_1,\omega_2,g)=(1.0,1.2,0.1)$ with the local controls $\sigma_x^{(1)}$, $\sigma_y^{(1)}$, $\sigma_x^{(2)}$\black{, and} $\sigma_y^{(2)}$. We study the magnetic field $\omega_1'$ varying within $[1.0-\pi/T,1.0+\pi/T]$ while $\omega_2'=1.2$. As shown in Fig.~\ref{fig:fig6}(a), the generalizability of RL is weaker, with reasonably low CR bounds seen in a smaller range $0.9\lesssim\omega_1'\lesssim1.3$ as a direct result of limited control channels.

In the second situation, the coupling strength $g$ \black{changes and} the analytical scheme no longer holds because the change $\delta H{(g,g')}\sim \sigma_z^{(1)}\sigma_z^{(2)}$ \black{cannot} be simply decomposed into a combination of local controls, even with all control channels in Eq.~\eqref{eq:localc} available. \black{Figure}~\ref{fig:fig6}(b) shows the CR bound under RL-designed controls as a function of $g'$, which varies \black{within} $[0.1-\pi/T,0.1+\pi/T]$. The RL agent is trained at $(\omega_1,\omega_2,g)=(1.0,1.2,0.1)$. The values of \black{the} CR bound from RL are much higher than the no-control case, negating the generalizability with respect to the \black{$ZZ$-coupling} strength. Smaller values of \black{the} CR bound are attained by GRAPE, which means that the control fields now must \black{be} optimized for every individual $g'$ using \black{either} GRAPE or \black{RL} but with the agent \black{retrained} at each individual new $g'$ value. Practically, we note that quantum parameter estimation is not the only technique to measure the coupling strength. Taking the coupling between nuclear spins in molecules as an example, its magnitude can be determined by common splitting of spins and its sign can be found by spin-selective pluses \cite{Vandersypen2005NMR}. In other words, even if our generalizable quantum parameter estimation fails for the coupling between spins, our results on the strength of magnetic field $\omega_k$ are more relevant to experimental implementation.

\section{Conclusion}\label{sec:concl}

In this \black{work we} discussed the generalizability of the \black{multiparameter} quantum estimation schemes in two representative two-qubit systems in \black{the} presence of dephasing. \black{Generalizability} refers to a characteristic of a method that can derive from a known (initial) optimal control for a given set of parameters to the optimal control for another set of parameters, with minimal cost. The initial optimal control can be found \black{by either} RL or GRAPE. When the controls are \black{complete}, \black{i.e.,}~the shifts in the Hamiltonian due to change in parameters can be completely absorbed in the available control channels, one may generalize the optimal control using \black{an} analytical method. On the other hand, when the controls are limited such that these shifts cannot be completely absorbed in the available control channels, such as the case in which the control of a spin system is restricted to certain directions while the evolution of the spins does not have the restriction, the analytical scheme is invalid. However, RL retains a level of generalizability, albeit for a narrower range. \black{Finally}, if it is impossible to decompose the shift in the Hamiltonian into available controls, such as the \black{$ZZ$ coupling} of two spins in Example 2 considered here, neither RL nor the analytical scheme provides any generalizability and one must spend a \black{great deal} of resources to rerun either RL or GRAPE for each new set of parameters encountered. The analytical scheme provides important insight \black{into} when and how the generalizable optimal control of RL may be implemented, as the way RL generalizes controls is found to be similar to the analytical scheme. Our results should facilitate efficient implementation of optimal quantum estimation of multiple parameters, particularly for an ensemble of systems with ranges of parameters.

\section*{Acknowledgements}  This work \black{was} supported by the Key-Area Research and Development Program of GuangDong Province  (Grant No.~2018B030326001), the National Natural Science Foundation of China (Grant No.~11874312), the Research Grants Council of Hong Kong (\black{Grants No.} CityU 11303617, \black{No.} CityU 11304018, \black{No.} CityU 11304920, \black{and No.} CUHK 14308019), the Research Strategic Funding Scheme of CUHK (Grant No. 3133234), and the Guangdong Innovative and Entrepreneurial Research Team Program (Grant No. 2016ZT06D348).

\appendix

\section{Practical setup of RL}\label{sec:appx-ddpg}
In \black{the} DDPG \black{algorithm}, the RL agent consists of a deep-$Q$ network denoted by $Q^{\mu}(s,a|\theta^{Q})$ and a (deterministic) policy network denoted by $\mu(s|\theta)$\black{, which} are parameterized by $\theta^{Q}$ and $\theta$ respectively. \black{Here the} state-action value function $Q^{\mu}\qty(s_j,a_j)=\mathbb{E}\qty[R_j|s_j,a_j]$ describes the expected total discounted reward of state $s_j$ for taking the action $a_j$ \cite{sutton2018reinforcement}.
The policy network has two hidden layers with 400 and 300 neurons, respectively. The deep-$Q$ network has two hidden layers with 400 and \black{300+$P$} neurons, respectively. \black{Here} $P$ represents the dimension of the action space\black{, i.e.,} the number of control channels. Note that the second hidden layer in deep-$Q$ network concatenates the data from the output of its first hidden layer and the actions from the policy network. For both networks, the input layers take a 32-dimensional real vector corresponding to the fully known $4\times4$ density matrix $\rho(j\Delta t)$. The \black{$\text{elu}$} activation function \cite{PyTorch2019} is used after their hidden layers, and the final output layer of policy network uses \black{$\text{tanh}$} \cite{PyTorch2019} to bound the actions.

The DDPG algorithm adapts two sets of neural networks similar to the \black{deep-$Q$ learning}: exploring the RL environment with \black{one} set of networks and training the other set (\black{namely,} the target networks) using the samples in a replay memory \cite{mnih2013playing,lillicrap2015continuous}. When performing the update to neural networks, we uniformly choose random samples from the replay memory so that the correlation between successive updates is removed, which reduces the variance of learning \cite{sutton2018reinforcement}.
The Adam algorithm is used for the target deep-$Q$ and target policy networks with the learning rate $10^{-4}$ \cite{PyTorch2019}. Other hyperparameters are identical to the original DDPG paper \cite{lillicrap2015continuous}.

In order to prevent the agent from learning undesirable solutions \cite{sutton2018reinforcement}, we define the reward function that has \black{a} \black{nonzero} value only at the target time $T$,
\begin{equation}\label{eq:reward}
    r_j= \begin{cases}
        0 & \black{\text{for\quad}} j\ne N, \\
        100\sum_{n=1}^{4}10^{-10^{n}\trcfiminv} & \black{\text{for\quad}} j=N,
        \end{cases}
\end{equation}
where the exponential encourages the agent to reduce $\trcfiminv$.

\section{Training procedure}\label{sec:appx-hyper}

\begin{figure}[tb]
    \includegraphics[width=0.7\linewidth]{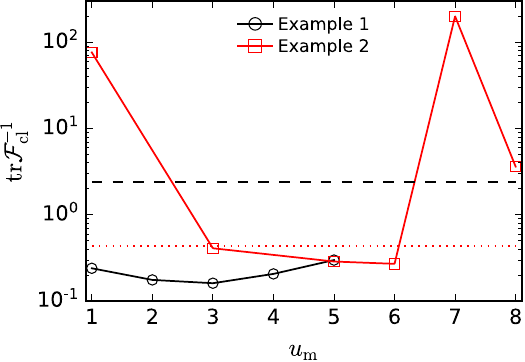}
    \caption{Performance of RL using different ranges of action space. We show the value of $\trcfiminv$ as a function of the range $u_\mathrm{m}$ \black{where} $|{u_a^{(i)}}|\leqslant u_\mathrm{m}$. \black{The other parameters are $T=5$ and $\Delta t=0.1$.} The same random seed is used. The black dashed and red dotted horizontals are the CR bound without control in \black{Examples 1 and 2}, respectively.}
    \label{fig:fig7}
\end{figure}

Training a neural network is usually time consuming, demanding heavy tuning of hyperparameters.  In this \black{section we} investigate the performance of RL under different choices of the hyperparameters and exemplify the RL and GRAPE methods regarding their time complexity. The following simulations use the parameters $(B,\vartheta,\varphi)=(1,\pi/4,\pi/4)$ for Example 1 and $(\omega_1,\omega_2,g)=(1.0,1.2,0.1)$ for Example 2.

\black{First}, the controls in our simulation are \black{piecewise} square pulses where the strengths of controls, $u_a^{(i)}$, are directly relevant to the performance of RL. The \black{range $|{u_a^{(i)}}|\leqslant u_\mathrm{m}$ significantly} affects the efficiency and stability of the RL procedure.
In \black{Fig.~\ref{fig:fig7} the} results show that the range of control has little effect on $\trcfiminv$ at target time $T$ in Example 1, while the range can significantly affect the performance of RL in Example 2. From these results, we set the range as $|{u_a^{(i)}}|\leqslant3$ in Example 1 and $|{u_a^{(i)}}|\leqslant5$ in Example 2.

\black{Second}, we compare the training of RL using 2000, \black{10\,000, and 20\,000} RL episodes, with results shown in Fig.~\ref{fig:fig8}. For both Example 1 and \black{Example 2 we} have taken $T=5$ and $\Delta t=0.1$\black{,} so there are $50$ state-action pairs in each RL episode. The size of \black{the} replay memory is typically chosen as \black{one-tenth} of the total state-action pairs \cite{mnih2013playing}. The results show that in both cases, RL with \black{fewer} than \black{10\,000} RL episodes fails to learn a good policy. In order to make the training satisfactory, we use \black{10\,000} RL episodes with replay memory size \black{of} \black{50\,000} in Example 1 and use \black{20\,000} RL episodes with replay memory size \black{100\,000} in Example 2. When $T$ is larger than $5$, longer RL training episodes and larger replay memory should be adapted accordingly.

\begin{figure}[tb]
  \includegraphics[width=0.8\linewidth]{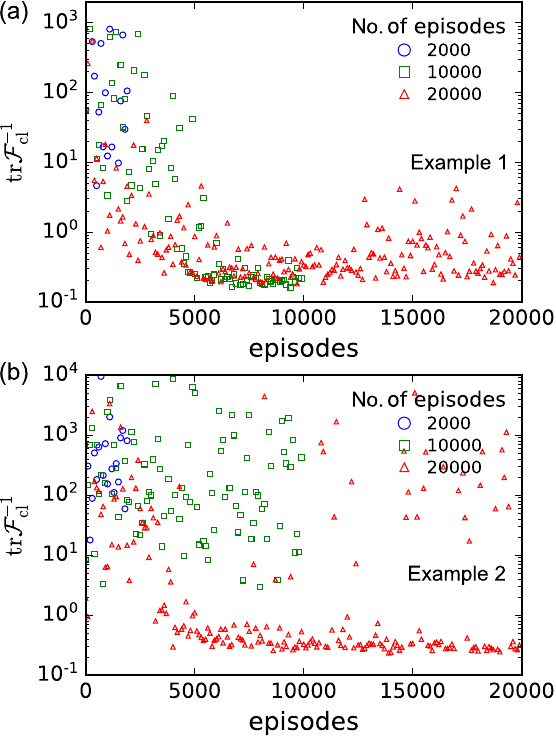}
  \caption{Training processes of RL in (a) Example 1 and (b) Example 2. \black{Here} $T=5$ \black{and} $\Delta t=0.1$. Blue marks represent results from a total of 2000 RL episodes with replay memory size \black{10\,000}, green marks correspond to a training involving \black{10\,000} RL episodes with replay memory size \black{20\,000}, and red marks represent \black{20\,000} RL episodes with replay memory size \black{100\,000}.}
  \label{fig:fig8}
\end{figure}

\begin{figure}[tb]
  \includegraphics[width=0.8\linewidth]{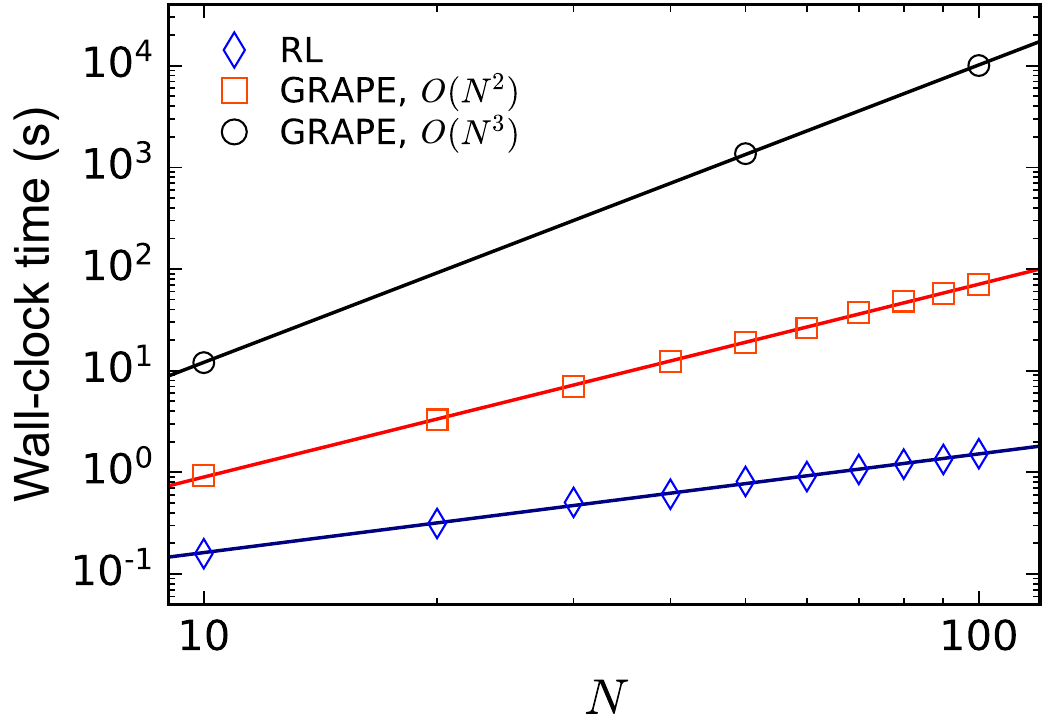}
  \caption{Time complexity of different methods. The \black{$x$ axis} shows the system size, i.e., the number of time slices $N$. The \black{$y$ axis} shows the wall-clock time cost in seconds. Diamonds on the blue line show the time cost in one RL training episode. Squares on the red line and circles on the black line are the \black{times} required by one iteration of GRAPE with and without storing the superoperators, respectively. Notice that the base-10 \black{logarithmic} scale is used for both axes.}\label{fig:fig9}
\end{figure}

\begin{figure}[tb]
    \includegraphics[width=0.8\linewidth]{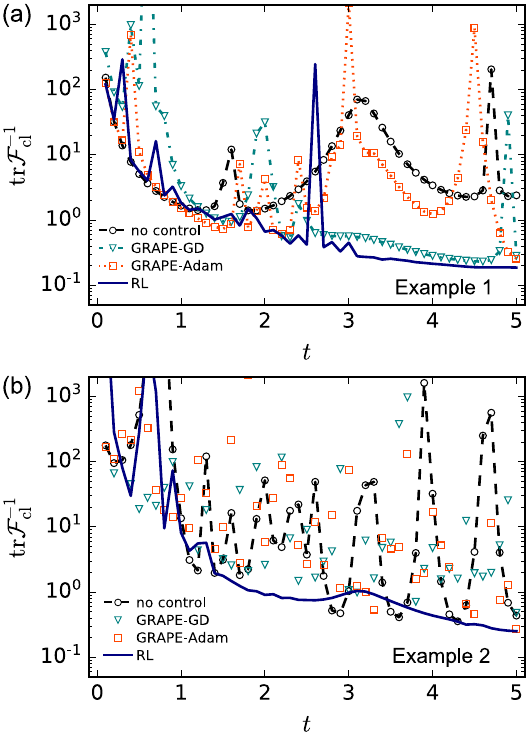}
    \caption{\black{Classical} CR \black{bound $\trcfiminv$} as \black{a function} of $t$ in (a) Example 1 and (b) Example 2. The target time $T=5$. The parameters are $(B,\vartheta,\varphi)=(1,\pi/4,\pi/4)$ in Example 1 and $(\omega_1,\omega_2,g)=(1.0,1.2,0.1)$ in Example 2. The dashed black curves show the results without control. The solid blue curves show the results from RL. The dash-dotted green curves and the dotted orange curves represent results obtained by GRAPE with \black{the} steepest gradient ascent (GRAPE-GD) and Adam (GRAPE-Adam) methods, respectively. For GRAPE, the initial guess of the control is zero in Example 1 and the initial guess is random in Example 2.}\label{fig:fig10}
\end{figure}

\begin{figure}[t]
    \includegraphics[width=0.8\linewidth]{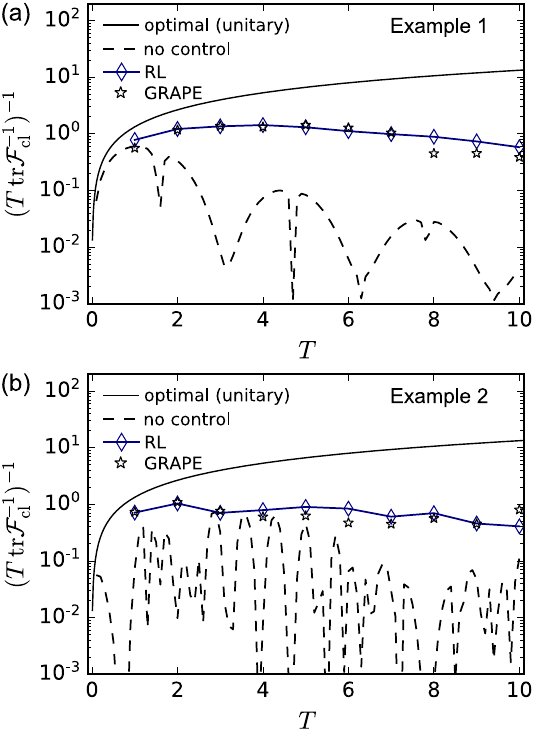}
    \caption{\black{Plot of} $\qty(T\trcfiminv)^{-1}$ as \black{a function} of the target time $T$ in (a) Example 1 \black{with $(B,\vartheta,\varphi)=(1,\pi/4,\pi/4)$} and (b) Example 2 with \black{$(\omega_1,\omega_2,g)=(1.0,1.2,0.1)$}. Diamonds on the blue lines show results obtained by RL. Stars located at each $T$ are  GRAPE results from Ref.~\cite{liu2017control}. The dashed black curves show the value without control. The solid black curves represent the optimal precision limit $4T/3$ under unitary dynamics.}
    \label{fig:fig11}
\end{figure}

\begin{figure}[!]
  \includegraphics[width=0.8\linewidth]{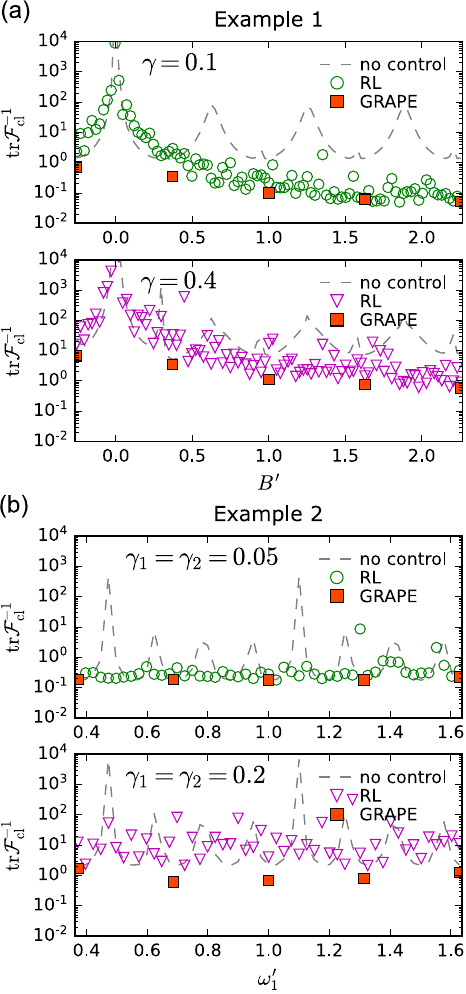}
  \caption{(a) Generalizability regarding the strength of the magnetic field in Example 1. \black{Here} $\trcfiminv$ evaluated at the target time $T=5$ \black{is} shown as \black{a function} of $B'$ in $[B-2\pi/T,B+2\pi/T]$ where $B=1$. The RL agent is trained at $(B,\vartheta,\varphi)=(1,\pi/4,\pi/4)$ and the dephasing rate $\gamma=0.2$, which is subsequently used to generate control pulses for $B'$ under dephasing rates $\gamma=0.1$ (upper panel) and \black{$\gamma=0.4$} (lower panel). (b) Generalizability regarding  individual magnetic fields in Example 2. The CR bounds evaluated at the target time $T=5$ are shown as \black{a function} of $\omega_1'$ with fixed $\omega_2=1.2$ and coupling strength $g=0.1$. The RL agent is trained at $(\omega_1,\omega_2,g)=(1.0,1.2,0.1)$ and dephasing rates $\gamma_{1,2}=0.1$, which \black{are} subsequently used to generate control pulses for a range of $\omega_1'$ under dephasing rates $\gamma_{1,2}=0.05$ (upper panel) and \black{$\gamma_{1,2}=0.2$} (lower panel). The dashed gray curves show the CR bound without control for comparison.}
  \label{fig:fig12}
\end{figure}

\black{Third}, we discuss the time complexity of different methods as a function of $N$, the total number of time slices. As shown in Fig.~\ref{fig:fig9}, the RL method scales as \black{$O{(N)}$}. For the GRAPE methods,  the slower \black{one, $O{(N^3)}$,} has repeatedly calculated the master equation, but the faster \black{one, $O{(N^2)}$,} uses more computer memory to store the superoperators. In general, RL requires about $10^4$ episodes to train an RL \black{agent,} but GRAPE needs about $10^2$ iterations to optimize a control field. For $T=5$ and $\Delta t=0.1$ used in this paper, $N=50$ and the total wall-clock time cost for RL and GRAPE to find the optimal control for a given set of parameters is actually comparable. We note that recent work suggests that the quantum version of \black{the DDPG algorithm} may offer \black{a} further \black{speedup} \cite{wu.2020}.

\section{Implementation of RL}\label{sec:irl}

In this \black{section we} explain \black{the} details of the implementation of RL and compare the optimization from RL and GRAPE \black{for} both examples shown in Sec.~\ref{sec:model}.
While we largely follows Ref.~\cite{liu2017control} for GRAPE, we additionally apply the steepest gradient ascent method \cite{ruder2016overview} (results shown as GRAPE-GD)  and the Adam method \cite{kingma2014adam} (results shown as GRAPE-Adam) to update the control field, which are numerically more efficient than the simple updating scheme.
For RL, the state $s_j$ is a 32-dimensional vector corresponding to real and imaginary parts of elements in the $4\times4$ density matrix $\rho(j\Delta t)$. The action $a_j$ is a \black{six}-dimensional vector \black{consisting} of elements in the control fields $\bm u^{(1)}(j\Delta t)$ and $\bm u^{(2)}(j\Delta t)$ on both qubits. For $j<N-1$, the immediate rewards are $r_{j+1}=0$, while $r_{N}$ is directly related to the CR \black{bound $\trcfiminv$} (\black{a} detailed form is shown in Appendix~\ref{sec:appx-ddpg}). In the RL procedure, the strength of the control field has \black{a} substantial impact on the outcome. A range is therefore imposed on the elements of control fields $u_a^{(i)}(t)\ (i=1,2)$ to avoid instability in the learning process. This range is set empirically as $|u_a^{(i)}|\leqslant3$ for Example 1 and $|u_a^{(i)}|\leqslant5$ for Example 2 as explained in Appendix~\ref{sec:appx-hyper}. No other restrictions on the RL procedure \black{are} imposed.

In Fig.~\ref{fig:fig10}\black{(a) we} show the values of the CR bound evolving from $t=0$ to the target time $t=T$ for Example 1 under quantum control. Overall, the CR bound is highest in the case of no control, indicating a low precision of the estimation. While all three methods produce lower CR bounds at the target time, spikes occur during the evolution using the GRAPE method, especially immediately before the target time. Instead, the precision under RL steadily saturates to the lowest value as the target time is reached. Results from RL is more desirable as in practice the time to implement the measurement may not be accurate \cite{durkin2010preferred,degen2017sensing} \black{and} spikes in the CR bounds could result in a rather low precision when the measurement is not exerted precisely at the targeted time.

The time evolution of the CR bound for Example 2 is shown in Fig.~\ref{fig:fig10}(b). Again, results from RL steadily reach the lowest value, while results from all other methods are unstable over the course of the evolution. An interesting observation is that the result with no control is comparable to other methods at the target time, which may be due to the local measurement preventing all methods from reaching asymptotically optimal solutions \cite{liu2017control,liu2019quantum,Humphreys2013Quantum,szczykulska2016multi}.

The CR bound under controls obtained from GRAPE fluctuates over the course of evolution because the control field is only optimized to minimize the bound at the final time. The RL agent, however, is trained to use samples from the replay \black{memory,} which breaks temporal correlation in each episode. In other words, the RL agent is trained to gradually steer states towards the final state. Therefore, as the system evolves close to the target time, the state is already very close to the final state, which is optimal for the estimation \cite{liu2017quantum, liu2017control, xu2019generalizable}. This leads to less \black{fluctuation} in the \black{results,} as shown in Fig.~\ref{fig:fig10}. Therefore, although both RL and GRAPE are capable \black{of finding} optimal controls at the target time, results from RL have less \black{fluctuation} when the target time is approached, implying \black{a} certain level of robustness against the fluctuations in the time to exert the measurement.

Assuming that the target time is precisely known, we compare different methods on the optimized CR bound as a function of the final time $T$, namely, we allow $T$ to \black{vary and} use the reciprocal of \black{the} normalized CR bound, $\qty(T\trcfiminv)^{-1}$, for comparison.
Under the unitary dynamics, the Heisenberg scaling can be attained for both examples \cite{yuan2016sequential,liu2017control},  corresponding to $\qty(T\trcfiminv)^{-1}=4T/3$. In \black{the} presence of noises, $\qty(T\trcfiminv)^{-1}$ typically increases at the beginning of the evolution and decreases thereafter \cite{liu2017control}, suggesting that the averaged \black{profit} (each time unit) on the precision limit from the time evolution is maximal at a finite \black{$T$,} which is the optimal time to perform the measurement.

For Example 1, Fig.~\ref{fig:fig11}(a) shows the results of $\qty(T\trcfiminv)^{-1}$ \black{vs} $T$. As can be seen from the figure, the maximal value without control is achieved around $T\approx 1.2$, while the largest $\qty(T\trcfiminv)^{-1}$ obtained by RL occurs around $3\lesssim T\lesssim5$, where the effect from control is maximized. As $T$ further increases, the accumulated dephasing effect becomes dominant and the controls may not optimally counteract the noise.
For Example 2, Fig.~\ref{fig:fig11}(b) shows that without control, $\qty(T\trcfiminv)^{-1}$ attains the largest value around $T\approx 3$, but under RL the value of $\qty(T\trcfiminv)^{-1}$ fluctuates very little \black{for} $4\lesssim T\lesssim 6$, and \black{tends} to decrease for $T\gtrsim6$. To ensure consistency, we have fixed $T=5$ for both examples in the main text.
We also present GRAPE results using the same method of Ref.~\cite{liu2017control}. As shown in Fig.~\ref{fig:fig11}, RL and GRAPE methods give comparable results for most values of $T$, with small differences at certain $T$ values.

\section{Generalizability under other values of the dephasing rates}\label{sec:appx-rate}
In this appendix we study the generalizability of RL associated \black{with} varying values of the dephasing rates that are applied in both Example 1 and Example 2.

\black{Figure}~\ref{fig:fig12}(a) shows the CR bounds in Example 1 obtained as \black{a function} of varying magnetic field strength $B'$ under different dephasing rates as indicated. The RL agent is trained at $(B,\vartheta,\varphi)=(1,\pi/4,\pi/4)$ and the dephasing rate $\gamma=0.2$, which is subsequently used to  generate control pulses for $B'$ under dephasing rates $\gamma=0.1$ (upper panel) and \black{$\gamma=0.4$} (lower panel). 
For comparison, the CR bounds obtained by rerunning GRAPE at selected $B'$ values are \black{represented by} squares, while the dashed curves show the CR bounds without control.  As the dephasing rate becomes larger, the CR bound increases, indicating a lower precision due to enlarged noise as expected.
For the smaller dephasing rate $\gamma=0.1$, \black{the} results from RL are comparable to \black{those from} GRAPE for $B'>1$, even if the agent is not trained specifically at those values of $B'$, indicating a good level of generalizability. For the larger dephasing rate $\gamma=0.4$, RL can no longer generate the pulse sequences as \black{optimally} as GRAPE, but the RL results are generally lower than the case without control for a wide range around $B'=1$, suggesting that the generalizability is retained to \black{a} certain extent for larger dephasing rates.
We emphasize that, for $|B'|\sim 0$, no method provides \black{a} notable improvement as the optimal solution diverges according to Eq.~\eqref{eq:gene_B}.

\black{Figure}~\ref{fig:fig12}(b) shows the CR bounds in Example 2 obtained as \black{a} function of varying $\omega_1'$ with $\omega_2'$ fixed at 1.2 under different dephasing rates as indicated.  The RL agent is trained at $(\omega_1,\omega_2,g)=(1.0,1.2,0.1)$ and the dephasing rate $\gamma_{1}=\gamma_{2}=0.1$, which is subsequently used to  generate control pulses under dephasing rates $\gamma_{1}=\gamma_{2}=0.05$ (upper panel) and $\gamma_{1}=\gamma_{2}=0.2$ (lower panel). 
For the smaller dephasing rate $\gamma_{1,2}=0.05$, results from RL are similar to those from GRAPE, demonstrating a high level of generalizability as compared to the no-control case. On the other hand, for the larger dephasing rate $\gamma_{1,2}=0.2$, RL does not offer \black{a} notable improvement from the case without control and is inferior to GRAPE, implying that the generalizability is compromised by large noises. Enhancement of the ability of RL to accommodate large noises would be an important topic for further study, which would likely require more sophisticated techniques such as \black{domain randomization} \cite{Tobin2017domain} and \black{network randomization} \cite{lee2019network}.


\begin{thebibliography}{47}%
    \makeatletter
    \providecommand \@ifxundefined [1]{%
     \@ifx{#1\undefined}
    }%
    \providecommand \@ifnum [1]{%
     \ifnum #1\expandafter \@firstoftwo
     \else \expandafter \@secondoftwo
     \fi
    }%
    \providecommand \@ifx [1]{%
     \ifx #1\expandafter \@firstoftwo
     \else \expandafter \@secondoftwo
     \fi
    }%
    \providecommand \natexlab [1]{#1}%
    \providecommand \enquote  [1]{``#1''}%
    \providecommand \bibnamefont  [1]{#1}%
    \providecommand \bibfnamefont [1]{#1}%
    \providecommand \citenamefont [1]{#1}%
    \providecommand \href@noop [0]{\@secondoftwo}%
    \providecommand \href [0]{\begingroup \@sanitize@url \@href}%
    \providecommand \@href[1]{\@@startlink{#1}\@@href}%
    \providecommand \@@href[1]{\endgroup#1\@@endlink}%
    \providecommand \@sanitize@url [0]{\catcode `\\12\catcode `\$12\catcode
      `\&12\catcode `\#12\catcode `\^12\catcode `\_12\catcode `\%12\relax}%
    \providecommand \@@startlink[1]{}%
    \providecommand \@@endlink[0]{}%
    \providecommand \url  [0]{\begingroup\@sanitize@url \@url }%
    \providecommand \@url [1]{\endgroup\@href {#1}{\urlprefix }}%
    \providecommand \urlprefix  [0]{URL }%
    \providecommand \Eprint [0]{\href }%
    \providecommand \doibase [0]{http://dx.doi.org/}%
    \providecommand \selectlanguage [0]{\@gobble}%
    \providecommand \bibinfo  [0]{\@secondoftwo}%
    \providecommand \bibfield  [0]{\@secondoftwo}%
    \providecommand \translation [1]{[#1]}%
    \providecommand \BibitemOpen [0]{}%
    \providecommand \bibitemStop [0]{}%
    \providecommand \bibitemNoStop [0]{.\EOS\space}%
    \providecommand \EOS [0]{\spacefactor3000\relax}%
    \providecommand \BibitemShut  [1]{\csname bibitem#1\endcsname}%
    \let\auto@bib@innerbib\@empty
    \bibitem [{\citenamefont {Humphreys}\ \emph {et~al.}(2013)\citenamefont
      {Humphreys}, \citenamefont {Barbieri}, \citenamefont {Datta},\ and\
      \citenamefont {Walmsley}}]{Humphreys2013Quantum}%
      \BibitemOpen
      \bibfield  {author} {\bibinfo {author} {\bibfnamefont {P.~C.}\ \bibnamefont
      {Humphreys}}, \bibinfo {author} {\bibfnamefont {M.}~\bibnamefont {Barbieri}},
      \bibinfo {author} {\bibfnamefont {A.}~\bibnamefont {Datta}}, \ and\ \bibinfo
      {author} {\bibfnamefont {I.~A.}\ \bibnamefont {Walmsley}},\ }\href {\doibase
      10.1103/PhysRevLett.111.070403} {\bibfield  {journal} {\bibinfo  {journal}
      {Phys. Rev. Lett.}\ }\textbf {\bibinfo {volume} {111}},\ \bibinfo {pages}
      {070403} (\bibinfo {year} {2013})}\BibitemShut {NoStop}%
    \bibitem [{\citenamefont {Pezz\`e}\ \emph {et~al.}(2017)\citenamefont
      {Pezz\`e}, \citenamefont {Ciampini}, \citenamefont {Spagnolo}, \citenamefont
      {Humphreys}, \citenamefont {Datta}, \citenamefont {Walmsley}, \citenamefont
      {Barbieri}, \citenamefont {Sciarrino},\ and\ \citenamefont
      {Smerzi}}]{pezze2017optimal}%
      \BibitemOpen
      \bibfield  {author} {\bibinfo {author} {\bibfnamefont {L.}~\bibnamefont
      {Pezz\`e}}, \bibinfo {author} {\bibfnamefont {M.~A.}\ \bibnamefont
      {Ciampini}}, \bibinfo {author} {\bibfnamefont {N.}~\bibnamefont {Spagnolo}},
      \bibinfo {author} {\bibfnamefont {P.~C.}\ \bibnamefont {Humphreys}}, \bibinfo
      {author} {\bibfnamefont {A.}~\bibnamefont {Datta}}, \bibinfo {author}
      {\bibfnamefont {I.~A.}\ \bibnamefont {Walmsley}}, \bibinfo {author}
      {\bibfnamefont {M.}~\bibnamefont {Barbieri}}, \bibinfo {author}
      {\bibfnamefont {F.}~\bibnamefont {Sciarrino}}, \ and\ \bibinfo {author}
      {\bibfnamefont {A.}~\bibnamefont {Smerzi}},\ }\href {\doibase
      10.1103/PhysRevLett.119.130504} {\bibfield  {journal} {\bibinfo  {journal}
      {Phys. Rev. Lett.}\ }\textbf {\bibinfo {volume} {119}},\ \bibinfo {pages}
      {130504} (\bibinfo {year} {2017})}\BibitemShut {NoStop}%
    \bibitem [{\citenamefont {Yue}\ \emph {et~al.}(2014)\citenamefont {Yue},
      \citenamefont {Zhang},\ and\ \citenamefont {Fan}}]{yue2014quantum}%
      \BibitemOpen
      \bibfield  {author} {\bibinfo {author} {\bibfnamefont {J.-D.}\ \bibnamefont
      {Yue}}, \bibinfo {author} {\bibfnamefont {Y.-R.}\ \bibnamefont {Zhang}}, \
      and\ \bibinfo {author} {\bibfnamefont {H.}~\bibnamefont {Fan}},\ }\href
      {https://www.nature.com/articles/srep05933} {\bibfield  {journal} {\bibinfo
      {journal} {Sci. Rep.}\ }\textbf {\bibinfo {volume} {4}},\ \bibinfo {pages}
      {5933} (\bibinfo {year} {2014})}\BibitemShut {NoStop}%
    \bibitem [{\citenamefont {Yuan}\ and\ \citenamefont
      {Fung}(2017)}]{yuan2017quantum}%
      \BibitemOpen
      \bibfield  {author} {\bibinfo {author} {\bibfnamefont {H.}~\bibnamefont
      {Yuan}}\ and\ \bibinfo {author} {\bibfnamefont {C.-H.~F.}\ \bibnamefont
      {Fung}},\ }\href {https://www.nature.com/articles/s41534-017-0014-6}
      {\bibfield  {journal} {\bibinfo  {journal} {npj Quantum Inf.}\ }\textbf
      {\bibinfo {volume} {3}},\ \bibinfo {pages} {14} (\bibinfo {year}
      {2017})}\BibitemShut {NoStop}%
    \bibitem [{\citenamefont {Yuan}\ and\ \citenamefont
      {Fung}(2015)}]{yuan2015Optimal}%
      \BibitemOpen
      \bibfield  {author} {\bibinfo {author} {\bibfnamefont {H.}~\bibnamefont
      {Yuan}}\ and\ \bibinfo {author} {\bibfnamefont {C.-H.~F.}\ \bibnamefont
      {Fung}},\ }\href {https://link.aps.org/doi/10.1103/PhysRevLett.115.110401}
      {\bibfield  {journal} {\bibinfo  {journal} {Phys. Rev. Lett.}\ }\textbf
      {\bibinfo {volume} {115}},\ \bibinfo {pages} {110401} (\bibinfo {year}
      {2015})}\BibitemShut {NoStop}%
    \bibitem [{\citenamefont {Liu}\ and\ \citenamefont
      {Yuan}(2017{\natexlab{a}})}]{liu2017quantum}%
      \BibitemOpen
      \bibfield  {author} {\bibinfo {author} {\bibfnamefont {J.}~\bibnamefont
      {Liu}}\ and\ \bibinfo {author} {\bibfnamefont {H.}~\bibnamefont {Yuan}},\
      }\href {\doibase 10.1103/PhysRevA.96.012117} {\bibfield  {journal} {\bibinfo
      {journal} {Phys. Rev. A}\ }\textbf {\bibinfo {volume} {96}},\ \bibinfo
      {pages} {012117} (\bibinfo {year} {2017}{\natexlab{a}})}\BibitemShut
      {NoStop}%
    \bibitem [{\citenamefont {Pang}\ and\ \citenamefont
      {Jordan}(2017)}]{pang2017optimal}%
      \BibitemOpen
      \bibfield  {author} {\bibinfo {author} {\bibfnamefont {S.}~\bibnamefont
      {Pang}}\ and\ \bibinfo {author} {\bibfnamefont {A.}~\bibnamefont {Jordan}},\
      }\href {https://dx.doi.org/10.1038/ncomms14695} {\bibfield  {journal}
      {\bibinfo  {journal} {Nat. Commun.}\ }\textbf {\bibinfo {volume} {8}},\
      \bibinfo {pages} {14695} (\bibinfo {year} {2017})}\BibitemShut {NoStop}%
    \bibitem [{\citenamefont {Fra\"{\i}sse}\ and\ \citenamefont
      {Braun}(2017)}]{fraisse2017enhancing}%
      \BibitemOpen
      \bibfield  {author} {\bibinfo {author} {\bibfnamefont {J.~M.~E.}\
      \bibnamefont {Fra\"{\i}sse}}\ and\ \bibinfo {author} {\bibfnamefont
      {D.}~\bibnamefont {Braun}},\ }\href
      {https://link.aps.org/doi/10.1103/PhysRevA.95.062342} {\bibfield  {journal}
      {\bibinfo  {journal} {Phys. Rev. A}\ }\textbf {\bibinfo {volume} {95}},\
      \bibinfo {pages} {062342} (\bibinfo {year} {2017})}\BibitemShut {NoStop}%
    \bibitem [{\citenamefont {Fiderer}\ \emph {et~al.}(2019)\citenamefont
      {Fiderer}, \citenamefont {Fra\"{\i}sse},\ and\ \citenamefont
      {Braun}}]{fiderer2019maximal}%
      \BibitemOpen
      \bibfield  {author} {\bibinfo {author} {\bibfnamefont {L.~J.}\ \bibnamefont
      {Fiderer}}, \bibinfo {author} {\bibfnamefont {J.~M.~E.}\ \bibnamefont
      {Fra\"{\i}sse}}, \ and\ \bibinfo {author} {\bibfnamefont {D.}~\bibnamefont
      {Braun}},\ }\href {https://link.aps.org/doi/10.1103/PhysRevLett.123.250502}
      {\bibfield  {journal} {\bibinfo  {journal} {Phys. Rev. Lett.}\ }\textbf
      {\bibinfo {volume} {123}},\ \bibinfo {pages} {250502} (\bibinfo {year}
      {2019})}\BibitemShut {NoStop}%
    \bibitem [{\citenamefont {Haine}\ and\ \citenamefont
      {Hope}(2020)}]{haine2020machine}%
      \BibitemOpen
      \bibfield  {author} {\bibinfo {author} {\bibfnamefont {S.~A.}\ \bibnamefont
      {Haine}}\ and\ \bibinfo {author} {\bibfnamefont {J.~J.}\ \bibnamefont
      {Hope}},\ }\href {https://link.aps.org/doi/10.1103/PhysRevLett.124.060402}
      {\bibfield  {journal} {\bibinfo  {journal} {Phys. Rev. Lett.}\ }\textbf
      {\bibinfo {volume} {124}},\ \bibinfo {pages} {060402} (\bibinfo {year}
      {2020})}\BibitemShut {NoStop}%
    \bibitem [{\citenamefont {Yuan}(2016)}]{yuan2016sequential}%
      \BibitemOpen
      \bibfield  {author} {\bibinfo {author} {\bibfnamefont {H.}~\bibnamefont
      {Yuan}},\ }\href {https://link.aps.org/doi/10.1103/PhysRevLett.117.160801}
      {\bibfield  {journal} {\bibinfo  {journal} {Phys. Rev. Lett.}\ }\textbf
      {\bibinfo {volume} {117}},\ \bibinfo {pages} {160801} (\bibinfo {year}
      {2016})}\BibitemShut {NoStop}%
    \bibitem [{\citenamefont {Liu}\ and\ \citenamefont
      {Yuan}(2017{\natexlab{b}})}]{liu2017control}%
      \BibitemOpen
      \bibfield  {author} {\bibinfo {author} {\bibfnamefont {J.}~\bibnamefont
      {Liu}}\ and\ \bibinfo {author} {\bibfnamefont {H.}~\bibnamefont {Yuan}},\
      }\href {https://link.aps.org/doi/10.1103/PhysRevA.96.042114} {\bibfield
      {journal} {\bibinfo  {journal} {Phys. Rev. A}\ }\textbf {\bibinfo {volume}
      {96}},\ \bibinfo {pages} {042114} (\bibinfo {year}
      {2017}{\natexlab{b}})}\BibitemShut {NoStop}%
    \bibitem [{\citenamefont {Khaneja}\ \emph {et~al.}(2005)\citenamefont
      {Khaneja}, \citenamefont {Reiss}, \citenamefont {Kehlet}, \citenamefont
      {Schulte-Herbr{\"u}ggen},\ and\ \citenamefont {Glaser}}]{Khaneja2005Optimal}%
      \BibitemOpen
      \bibfield  {author} {\bibinfo {author} {\bibfnamefont {N.}~\bibnamefont
      {Khaneja}}, \bibinfo {author} {\bibfnamefont {T.}~\bibnamefont {Reiss}},
      \bibinfo {author} {\bibfnamefont {C.}~\bibnamefont {Kehlet}}, \bibinfo
      {author} {\bibfnamefont {T.}~\bibnamefont {Schulte-Herbr{\"u}ggen}}, \ and\
      \bibinfo {author} {\bibfnamefont {S.~J.}\ \bibnamefont {Glaser}},\ }\href
      {\doibase 10.1016/j.jmr.2004.11.004} {\bibfield  {journal} {\bibinfo
      {journal} {J. Magn. Reson.}\ }\textbf {\bibinfo {volume} {172}},\ \bibinfo
      {pages} {296} (\bibinfo {year} {2005})}\BibitemShut {NoStop}%
    \bibitem [{\citenamefont {Mnih}\ \emph {et~al.}(2013)\citenamefont {Mnih},
      \citenamefont {Kavukcuoglu}, \citenamefont {Silver}, \citenamefont {Graves},
      \citenamefont {Antonoglou}, \citenamefont {Wierstra},\ and\ \citenamefont
      {Riedmiller}}]{mnih2013playing}%
      \BibitemOpen
      \bibfield  {author} {\bibinfo {author} {\bibfnamefont {V.}~\bibnamefont
      {Mnih}}, \bibinfo {author} {\bibfnamefont {K.}~\bibnamefont {Kavukcuoglu}},
      \bibinfo {author} {\bibfnamefont {D.}~\bibnamefont {Silver}}, \bibinfo
      {author} {\bibfnamefont {A.}~\bibnamefont {Graves}}, \bibinfo {author}
      {\bibfnamefont {I.}~\bibnamefont {Antonoglou}}, \bibinfo {author}
      {\bibfnamefont {D.}~\bibnamefont {Wierstra}}, \ and\ \bibinfo {author}
      {\bibfnamefont {M.}~\bibnamefont {Riedmiller}},\ }\href
      {https://arxiv.org/abs/1312.5602} {\bibfield  {journal} {\bibinfo  {journal}
      {arXiv:1312.5602}}}\BibitemShut
      {NoStop}%
    \bibitem [{\citenamefont {Mnih}\ \emph {et~al.}(2015)\citenamefont {Mnih},
      \citenamefont {Kavukcuoglu}, \citenamefont {Silver}, \citenamefont {Rusu},
      \citenamefont {Veness}, \citenamefont {Bellemare}, \citenamefont {Graves},
      \citenamefont {Riedmiller}, \citenamefont {Fidjeland}, \citenamefont
      {Ostrovski} \emph {et~al.}}]{mnih2015human}%
      \BibitemOpen
      \bibfield  {author} {\bibinfo {author} {\bibfnamefont {V.}~\bibnamefont
      {Mnih}}, \bibinfo {author} {\bibfnamefont {K.}~\bibnamefont {Kavukcuoglu}},
      \bibinfo {author} {\bibfnamefont {D.}~\bibnamefont {Silver}}, \bibinfo
      {author} {\bibfnamefont {A.~A.}\ \bibnamefont {Rusu}}, \bibinfo {author}
      {\bibfnamefont {J.}~\bibnamefont {Veness}}, \bibinfo {author} {\bibfnamefont
      {M.~G.}\ \bibnamefont {Bellemare}}, \bibinfo {author} {\bibfnamefont
      {A.}~\bibnamefont {Graves}}, \bibinfo {author} {\bibfnamefont
      {M.}~\bibnamefont {Riedmiller}}, \bibinfo {author} {\bibfnamefont {A.~K.}\
      \bibnamefont {Fidjeland}}, \bibinfo {author} {\bibfnamefont {G.}~\bibnamefont
      {Ostrovski}},  \emph {et~al.},\ }\href {\doibase 10.1038/nature14236}
      {\bibfield  {journal} {\bibinfo  {journal} {Nature (London)}\ }\textbf {\bibinfo
      {volume} {518}},\ \bibinfo {pages} {529} (\bibinfo {year}
      {2015})}\BibitemShut {NoStop}%
    \bibitem [{\citenamefont {Silver}\ \emph {et~al.}(2016)\citenamefont {Silver},
      \citenamefont {Huang}, \citenamefont {Maddison}, \citenamefont {Guez},
      \citenamefont {Sifre}, \citenamefont {van~den Driessche}, \citenamefont
      {Schrittwieser}, \citenamefont {Antonoglou}, \citenamefont {Panneershelvam},
      \citenamefont {Lanctot}, \citenamefont {Dieleman}, \citenamefont {Grewe},
      \citenamefont {Nham}, \citenamefont {Kalchbrenner}, \citenamefont
      {Sutskever}, \citenamefont {Lillicrap}, \citenamefont {Leach}, \citenamefont
      {Kavukcuoglu}, \citenamefont {Graepel},\ and\ \citenamefont
      {Hassabis}}]{silver2016mastering}%
      \BibitemOpen
      \bibfield  {author} {\bibinfo {author} {\bibfnamefont {D.}~\bibnamefont
      {Silver}}, \bibinfo {author} {\bibfnamefont {A.}~\bibnamefont {Huang}},
      \bibinfo {author} {\bibfnamefont {C.~J.}\ \bibnamefont {Maddison}}, \bibinfo
      {author} {\bibfnamefont {A.}~\bibnamefont {Guez}}, \bibinfo {author}
      {\bibfnamefont {L.}~\bibnamefont {Sifre}}, \bibinfo {author} {\bibfnamefont
      {G.}~\bibnamefont {van~den Driessche}}, \bibinfo {author} {\bibfnamefont
      {J.}~\bibnamefont {Schrittwieser}}, \bibinfo {author} {\bibfnamefont
      {I.}~\bibnamefont {Antonoglou}}, \bibinfo {author} {\bibfnamefont
      {V.}~\bibnamefont {Panneershelvam}}, \bibinfo {author} {\bibfnamefont
      {M.}~\bibnamefont {Lanctot}}, \bibinfo {author} {\bibfnamefont
      {S.}~\bibnamefont {Dieleman}}, \bibinfo {author} {\bibfnamefont
      {D.}~\bibnamefont {Grewe}}, \bibinfo {author} {\bibfnamefont
      {J.}~\bibnamefont {Nham}}, \bibinfo {author} {\bibfnamefont {N.}~\bibnamefont
      {Kalchbrenner}}, \bibinfo {author} {\bibfnamefont {I.}~\bibnamefont
      {Sutskever}}, \bibinfo {author} {\bibfnamefont {T.}~\bibnamefont
      {Lillicrap}}, \bibinfo {author} {\bibfnamefont {M.}~\bibnamefont {Leach}},
      \bibinfo {author} {\bibfnamefont {K.}~\bibnamefont {Kavukcuoglu}}, \bibinfo
      {author} {\bibfnamefont {T.}~\bibnamefont {Graepel}}, \ and\ \bibinfo
      {author} {\bibfnamefont {D.}~\bibnamefont {Hassabis}},\ }\href
      {https://www.nature.com/articles/nature16961} {\bibfield  {journal} {\bibinfo
       {journal} {Nature (London)}\ }\textbf {\bibinfo {volume} {529}},\ \bibinfo {pages}
      {484} (\bibinfo {year} {2016})}\BibitemShut {NoStop}%
    \bibitem [{\citenamefont {F\"osel}\ \emph {et~al.}(2018)\citenamefont
      {F\"osel}, \citenamefont {Tighineanu}, \citenamefont {Weiss},\ and\
      \citenamefont {Marquardt}}]{Thomas2018Reinforcement}%
      \BibitemOpen
      \bibfield  {author} {\bibinfo {author} {\bibfnamefont {T.}~\bibnamefont
      {F\"osel}}, \bibinfo {author} {\bibfnamefont {P.}~\bibnamefont {Tighineanu}},
      \bibinfo {author} {\bibfnamefont {T.}~\bibnamefont {Weiss}}, \ and\ \bibinfo
      {author} {\bibfnamefont {F.}~\bibnamefont {Marquardt}},\ }\href {\doibase
      10.1103/PhysRevX.8.031084} {\bibfield  {journal} {\bibinfo  {journal} {Phys.
      Rev. X}\ }\textbf {\bibinfo {volume} {8}},\ \bibinfo {pages} {031084}
      (\bibinfo {year} {2018})}\BibitemShut {NoStop}%
    \bibitem [{\citenamefont {Bukov}\ \emph {et~al.}(2018)\citenamefont {Bukov},
      \citenamefont {Day}, \citenamefont {Sels}, \citenamefont {Weinberg},
      \citenamefont {Polkovnikov},\ and\ \citenamefont
      {Mehta}}]{Bukov2018Reinforcement}%
      \BibitemOpen
      \bibfield  {author} {\bibinfo {author} {\bibfnamefont {M.}~\bibnamefont
      {Bukov}}, \bibinfo {author} {\bibfnamefont {A.~G.~R.}\ \bibnamefont {Day}},
      \bibinfo {author} {\bibfnamefont {D.}~\bibnamefont {Sels}}, \bibinfo {author}
      {\bibfnamefont {P.}~\bibnamefont {Weinberg}}, \bibinfo {author}
      {\bibfnamefont {A.}~\bibnamefont {Polkovnikov}}, \ and\ \bibinfo {author}
      {\bibfnamefont {P.}~\bibnamefont {Mehta}},\ }\href {\doibase
      10.1103/PhysRevX.8.031086} {\bibfield  {journal} {\bibinfo  {journal} {Phys.
      Rev. X}\ }\textbf {\bibinfo {volume} {8}},\ \bibinfo {pages} {031086}
      (\bibinfo {year} {2018})}\BibitemShut {NoStop}%
    \bibitem [{\citenamefont {Niu}\ \emph {et~al.}(2019)\citenamefont {Niu},
      \citenamefont {Boixo}, \citenamefont {Smelyanskiy},\ and\ \citenamefont
      {Neven}}]{niu2019universal}%
      \BibitemOpen
      \bibfield  {author} {\bibinfo {author} {\bibfnamefont {M.~Y.}\ \bibnamefont
      {Niu}}, \bibinfo {author} {\bibfnamefont {S.}~\bibnamefont {Boixo}}, \bibinfo
      {author} {\bibfnamefont {V.~N.}\ \bibnamefont {Smelyanskiy}}, \ and\ \bibinfo
      {author} {\bibfnamefont {H.}~\bibnamefont {Neven}},\ }\href
      {https://www.nature.com/articles/s41534-019-0141-3} {\bibfield  {journal}
      {\bibinfo  {journal} {npj Quantum Inf.}\ }\textbf {\bibinfo {volume} {5}},\
      \bibinfo {pages} {33} (\bibinfo {year} {2019})}\BibitemShut {NoStop}%
    \bibitem [{\citenamefont {Zhang}\ \emph {et~al.}(2019)\citenamefont {Zhang},
      \citenamefont {Wei}, \citenamefont {Asad}, \citenamefont {Yang},\ and\
      \citenamefont {Wang}}]{zhang2019does}%
      \BibitemOpen
      \bibfield  {author} {\bibinfo {author} {\bibfnamefont {X.-M.}\ \bibnamefont
      {Zhang}}, \bibinfo {author} {\bibfnamefont {Z.}~\bibnamefont {Wei}}, \bibinfo
      {author} {\bibfnamefont {R.}~\bibnamefont {Asad}}, \bibinfo {author}
      {\bibfnamefont {X.-C.}\ \bibnamefont {Yang}}, \ and\ \bibinfo {author}
      {\bibfnamefont {X.}~\bibnamefont {Wang}},\ }\href
      {https://www.nature.com/articles/s41534-019-0201-8} {\bibfield  {journal}
      {\bibinfo  {journal} {npj Quantum Inf.}\ }\textbf {\bibinfo {volume} {5}},\
      \bibinfo {pages} {85} (\bibinfo {year} {2019})}\BibitemShut {NoStop}%
    \bibitem [{\citenamefont {Lin}\ \emph {et~al.}(2020)\citenamefont {Lin},
      \citenamefont {Lai},\ and\ \citenamefont {Li}}]{lin2020quantum}%
      \BibitemOpen
      \bibfield  {author} {\bibinfo {author} {\bibfnamefont {J.}~\bibnamefont
      {Lin}}, \bibinfo {author} {\bibfnamefont {Z.~Y.}\ \bibnamefont {Lai}}, \ and\
      \bibinfo {author} {\bibfnamefont {X.}~\bibnamefont {Li}},\ }\href
      {https://link.aps.org/doi/10.1103/PhysRevA.101.052327} {\bibfield  {journal}
      {\bibinfo  {journal} {Phys. Rev. A}\ }\textbf {\bibinfo {volume} {101}},\
      \bibinfo {pages} {052327} (\bibinfo {year} {2020})}\BibitemShut {NoStop}%
    \bibitem [{\citenamefont {Xu}\ \emph {et~al.}(2019)\citenamefont {Xu},
      \citenamefont {Li}, \citenamefont {Liu}, \citenamefont {Wang}, \citenamefont
      {Yuan},\ and\ \citenamefont {Wang}}]{xu2019generalizable}%
      \BibitemOpen
      \bibfield  {author} {\bibinfo {author} {\bibfnamefont {H.}~\bibnamefont
      {Xu}}, \bibinfo {author} {\bibfnamefont {J.}~\bibnamefont {Li}}, \bibinfo
      {author} {\bibfnamefont {L.}~\bibnamefont {Liu}}, \bibinfo {author}
      {\bibfnamefont {Y.}~\bibnamefont {Wang}}, \bibinfo {author} {\bibfnamefont
      {H.}~\bibnamefont {Yuan}}, \ and\ \bibinfo {author} {\bibfnamefont
      {X.}~\bibnamefont {Wang}},\ }\href
      {https://www.nature.com/articles/s41534-019-0198-z} {\bibfield  {journal}
      {\bibinfo  {journal} {npj Quantum Inf.}\ }\textbf {\bibinfo {volume} {5}},\
      \bibinfo {pages} {82} (\bibinfo {year} {2019})}\BibitemShut {NoStop}%
    \bibitem [{\citenamefont {Schuff}\ \emph {et~al.}(2020)\citenamefont {Schuff},
      \citenamefont {Fiderer},\ and\ \citenamefont {Braun}}]{schuff2020improving}%
      \BibitemOpen
      \bibfield  {author} {\bibinfo {author} {\bibfnamefont {J.}~\bibnamefont
      {Schuff}}, \bibinfo {author} {\bibfnamefont {L.~J.}\ \bibnamefont {Fiderer}},
      \ and\ \bibinfo {author} {\bibfnamefont {D.}~\bibnamefont {Braun}},\ }\href
      {https://iopscience.iop.org/article/10.1088/1367-2630/ab6f1f} {\bibfield
      {journal} {\bibinfo  {journal} {New J. Phys.}\ }\textbf {\bibinfo {volume}
      {22}},\ \bibinfo {pages} {035001} (\bibinfo {year} {2020})}\BibitemShut
      {NoStop}%
    \bibitem [{\citenamefont {Fiderer}\ \emph {et~al.}(2021)\citenamefont
      {Fiderer}, \citenamefont {Schuff},\ and\ \citenamefont
      {Braun}}]{fiderer2020neural}%
      \BibitemOpen
      \bibfield  {author} {\bibinfo {author} {\bibfnamefont {L.~J.}\ \bibnamefont
      {Fiderer}}, \bibinfo {author} {\bibfnamefont {J.}~\bibnamefont {Schuff}}, \
      and\ \bibinfo {author} {\bibfnamefont {D.}~\bibnamefont {Braun}},\ }\href
      {\doibase 10.1103/PRXQuantum.2.020303} {\bibfield  {journal} {\bibinfo
      {journal} {PRX Quantum}\ }\textbf {\bibinfo {volume} {2}},\ \bibinfo {pages}
      {020303} (\bibinfo {year} {2021})}\BibitemShut {NoStop}%
    \bibitem [{\citenamefont {Sutton}\ and\ \citenamefont
      {Barto}(2018)}]{sutton2018reinforcement}%
      \BibitemOpen
      \bibfield  {author} {\bibinfo {author} {\bibfnamefont {R.~S.}\ \bibnamefont
      {Sutton}}\ and\ \bibinfo {author} {\bibfnamefont {A.~G.}\ \bibnamefont
      {Barto}},\ }\href@noop {} {\emph {\bibinfo {title} {Reinforcement Learning:
      An Introduction}}}\ (\bibinfo  {publisher} {MIT Press, Cambridge},\ \bibinfo {year}
      {2018})\BibitemShut {NoStop}%
    \bibitem [{\citenamefont {{Tobin}}\ \emph {et~al.}(2017)\citenamefont
      {{Tobin}}, \citenamefont {{Fong}}, \citenamefont {{Ray}}, \citenamefont
      {{Schneider}}, \citenamefont {{Zaremba}},\ and\ \citenamefont
      {{Abbeel}}}]{Tobin2017domain}%
      \BibitemOpen
      \bibfield  {author} {\bibinfo {author} {\bibfnamefont {J.}~\bibnamefont
      {{Tobin}}}, \bibinfo {author} {\bibfnamefont {R.}~\bibnamefont {{Fong}}},
      \bibinfo {author} {\bibfnamefont {A.}~\bibnamefont {{Ray}}}, \bibinfo
      {author} {\bibfnamefont {J.}~\bibnamefont {{Schneider}}}, \bibinfo {author}
      {\bibfnamefont {W.}~\bibnamefont {{Zaremba}}}, \ and\ \bibinfo {author}
      {\bibfnamefont {P.}~\bibnamefont {{Abbeel}}},\ }in\ \href {\doibase
      10.1109/IROS.2017.8202133} {\emph {\bibinfo {booktitle} {Proceedings of
      the 2017 IEEE/RSJ
      International Conference on Intelligent Robots and Systems (IROS), Vancouver,
      2017}}}\
      (IEEE, Piscataway, \bibinfo {year} {2017})\ pp.\ \bibinfo {pages} {23--30}\BibitemShut
      {NoStop}%
    \bibitem [{\citenamefont {Lee}\ \emph {et~al.}(2020)\citenamefont {Lee},
      \citenamefont {Lee}, \citenamefont {Shin},\ and\ \citenamefont
      {Lee}}]{lee2019network}%
      \BibitemOpen
      \bibfield  {author} {\bibinfo {author} {\bibfnamefont {K.}~\bibnamefont
      {Lee}}, \bibinfo {author} {\bibfnamefont {K.}~\bibnamefont {Lee}}, \bibinfo
      {author} {\bibfnamefont {J.}~\bibnamefont {Shin}}, \ and\ \bibinfo {author}
      {\bibfnamefont {H.}~\bibnamefont {Lee}},\ }in\ \href
      {https://openreview.net/forum?id=HJgcvJBFvB} {\emph {\bibinfo {booktitle}
      {International Conference on Learning Representations (ICLR)}}}\ (\bibinfo
      {year} {2020})\BibitemShut {NoStop}%
    \bibitem [{sut()}]{sutton1996generalization}%
      \BibitemOpen
      \href@noop {} {}\bibinfo {note} {Sutton, Richard S,
      \href{http://papers.nips.cc/paper/1109-generalization-in-reinforcement-learning-successful-examples-using-sparse-coarse-coding.pdf}{\textit{Generalization
      in reinforcement learning: Successful examples using sparse coarse coding}}
      \rm{in} {\textit{Advances in Neural Information Processing Systems 8}},
      edited by D. S. Touretzky and M. C. Mozer and M. E. Hasselmo (MIT Press, Cambridge, 1996) pp. 1038--1044}\BibitemShut {NoStop}%
    \bibitem [{\citenamefont {Cobbe}\ \emph {et~al.}(2019)\citenamefont {Cobbe},
      \citenamefont {Klimov}, \citenamefont {Hesse}, \citenamefont {Kim},\ and\
      \citenamefont {Schulman}}]{Cobbe2018Quantifying}%
      \BibitemOpen
      \bibfield  {author} {\bibinfo {author} {\bibfnamefont {K.}~\bibnamefont
      {Cobbe}}, \bibinfo {author} {\bibfnamefont {O.}~\bibnamefont {Klimov}},
      \bibinfo {author} {\bibfnamefont {C.}~\bibnamefont {Hesse}}, \bibinfo
      {author} {\bibfnamefont {T.}~\bibnamefont {Kim}}, \ and\ \bibinfo {author}
      {\bibfnamefont {J.}~\bibnamefont {Schulman}},\ }in\ \href
      {http://proceedings.mlr.press/v97/cobbe19a.html} {\emph {\bibinfo {booktitle}
      {Proceedings of the 36th International Conference on Machine Learning}}},\
      \bibinfo {series} {Proceedings of Machine Learning Research}, Vol.~\bibinfo
      {volume} {97},\ \bibinfo {editor} {edited by\ \bibinfo {editor}
      {\bibfnamefont {K.}~\bibnamefont {Chaudhuri}}\ and\ \bibinfo {editor}
      {\bibfnamefont {R.}~\bibnamefont {Salakhutdinov}}}\ (\bibinfo  {publisher}
      {PMLR},\ \bibinfo {year} {2019})\ pp.\ \bibinfo {pages}
      {1282--1289}\BibitemShut {NoStop}%
    \bibitem [{\citenamefont {Zhang}\ \emph {et~al.}(2018)\citenamefont {Zhang},
      \citenamefont {Ballas},\ and\ \citenamefont {Pineau}}]{zhang2018dissection}%
      \BibitemOpen
      \bibfield  {author} {\bibinfo {author} {\bibfnamefont {A.}~\bibnamefont
      {Zhang}}, \bibinfo {author} {\bibfnamefont {N.}~\bibnamefont {Ballas}}, \
      and\ \bibinfo {author} {\bibfnamefont {J.}~\bibnamefont {Pineau}},\ }\href
      {https://arxiv.org/abs/1806.07937} {\bibfield  {journal} {\bibinfo  {journal}
      {arXiv:1806.07937}}}\BibitemShut
      {NoStop}%
    \bibitem [{\citenamefont {Packer}\ \emph {et~al.}(2019)\citenamefont {Packer},
      \citenamefont {Gao}, \citenamefont {Kos}, \citenamefont {Kr\"{a}henb\"{u}hl},
      \citenamefont {Koltun},\ and\ \citenamefont {Song}}]{packer2019assessing}%
      \BibitemOpen
      \bibfield  {author} {\bibinfo {author} {\bibfnamefont {C.}~\bibnamefont
      {Packer}}, \bibinfo {author} {\bibfnamefont {K.}~\bibnamefont {Gao}},
      \bibinfo {author} {\bibfnamefont {J.}~\bibnamefont {Kos}}, \bibinfo {author}
      {\bibfnamefont {P.}~\bibnamefont {Kr\"{a}henb\"{u}hl}}, \bibinfo {author}
      {\bibfnamefont {V.}~\bibnamefont {Koltun}}, \ and\ \bibinfo {author}
      {\bibfnamefont {D.}~\bibnamefont {Song}},\ }\href
      {https://arxiv.org/abs/1810.12282} {\bibfield  {journal} {\bibinfo  {journal}
      {arXiv:1810.12282}}}\BibitemShut
      {NoStop}%
    \bibitem [{\citenamefont {Pathak}\ \emph {et~al.}(2017)\citenamefont {Pathak},
      \citenamefont {Agrawal}, \citenamefont {Efros},\ and\ \citenamefont
      {Darrell}}]{Pathak2017}%
      \BibitemOpen
      \bibfield  {author} {\bibinfo {author} {\bibfnamefont {D.}~\bibnamefont
      {Pathak}}, \bibinfo {author} {\bibfnamefont {P.}~\bibnamefont {Agrawal}},
      \bibinfo {author} {\bibfnamefont {A.~A.}\ \bibnamefont {Efros}}, \ and\
      \bibinfo {author} {\bibfnamefont {T.}~\bibnamefont {Darrell}},\ }in\ \href
      {https://arxiv.org/abs/1705.05363} {\emph {\bibinfo {booktitle} {Proceedings
      of the IEEE Conference on Computer Vision and Pattern Recognition
      Workshops, Honolulu, 2017}}}\ (IEEE, Piscataway,
      \bibinfo {year} {2017}),\ pp.\ \bibinfo {pages}
      {16--17}\BibitemShut {NoStop}%
    \bibitem [{\citenamefont {Breuer}\ and\ \citenamefont
      {Petruccione}(2002)}]{breuer2002theory}%
      \BibitemOpen
      \bibfield  {author} {\bibinfo {author} {\bibfnamefont {H.-P.}\ \bibnamefont
      {Breuer}}\ and\ \bibinfo {author} {\bibfnamefont {F.}~\bibnamefont
      {Petruccione}},\ }\href@noop {} {\emph {\bibinfo {title} {The Theory of Open
      Quantum Systems}}}\ (\bibinfo  {publisher} {Oxford University Press, Oxford},\
      \bibinfo {year} {2002})\BibitemShut {NoStop}%
    \bibitem [{\citenamefont {Liu}\ \emph {et~al.}(2020)\citenamefont {Liu},
      \citenamefont {Yuan}, \citenamefont {Lu},\ and\ \citenamefont
      {Wang}}]{liu2019quantum}%
      \BibitemOpen
      \bibfield  {author} {\bibinfo {author} {\bibfnamefont {J.}~\bibnamefont
      {Liu}}, \bibinfo {author} {\bibfnamefont {H.}~\bibnamefont {Yuan}}, \bibinfo
      {author} {\bibfnamefont {X.-M.}\ \bibnamefont {Lu}}, \ and\ \bibinfo {author}
      {\bibfnamefont {X.}~\bibnamefont {Wang}},\ }\href
      {https://iopscience.iop.org/article/10.1088/1751-8121/ab5d4d} {\bibfield
      {journal} {\bibinfo  {journal} {J. Phys. A: Math. Theor.}\ }\textbf {\bibinfo {volume}
      {53}},\ \bibinfo {pages} {023001} (\bibinfo {year} {2020})}\BibitemShut
      {NoStop}%
    \bibitem [{\citenamefont {Ruder}(2016)}]{ruder2016overview}%
      \BibitemOpen
      \bibfield  {author} {\bibinfo {author} {\bibfnamefont {S.}~\bibnamefont
      {Ruder}},\ }\href {https://arxiv.org/abs/1609.04747} {\bibfield  {journal}
      {\bibinfo  {journal} {arXiv:1609.04747}}}\BibitemShut {NoStop}%
    \bibitem [{\citenamefont {Kingma}\ and\ \citenamefont
      {Ba}(2014)}]{kingma2014adam}%
      \BibitemOpen
      \bibfield  {author} {\bibinfo {author} {\bibfnamefont {D.~P.}\ \bibnamefont
      {Kingma}}\ and\ \bibinfo {author} {\bibfnamefont {J.}~\bibnamefont {Ba}},\
      }\href {https://arxiv.org/abs/1412.6980} {\bibfield  {journal} {\bibinfo
      {journal} {arXiv:1412.6980}}}\BibitemShut {NoStop}%
    \bibitem [{sil()}]{silver2014deterministic}%
      \BibitemOpen
      \href@noop {} {}\bibinfo {note} {D. Silver, G. Lever, N. Heess, T. Degris, D. Wierstra, and M. Riedmiller,
      \rm{in} \href{http://proceedings.mlr.press/v32/silver14.pdf} {\textit{Proceedings of the 31st
      International Conference on Machine Learning, Beijing, 2014}}, edited by Eric P. Xing and
      Tony Jebara (Microtome, Brookline, 2014) pp. 387--395}\BibitemShut {NoStop}%
    \bibitem [{\citenamefont {Lin}(1993)}]{lin1993reinforcement}%
      \BibitemOpen
      \bibfield  {author} {\bibinfo {author} {\bibfnamefont {L.-J.}\ \bibnamefont
      {Lin}},\ }\href@noop {} {\emph {\bibinfo {title} {Reinforcement learning for
      robots using neural networks}}},\ \bibinfo {type} {Tech. Rep.}\ (\bibinfo
      {institution} {Carnegie-Mellon Univ Pittsburgh PA School of Computer
      Science},\ \bibinfo {year} {1993})\BibitemShut {NoStop}%
    \bibitem [{\citenamefont {Lillicrap}\ \emph {et~al.}(2015)\citenamefont
      {Lillicrap}, \citenamefont {Hunt}, \citenamefont {Pritzel}, \citenamefont
      {Heess}, \citenamefont {Erez}, \citenamefont {Tassa}, \citenamefont
      {Silver},\ and\ \citenamefont {Wierstra}}]{lillicrap2015continuous}%
      \BibitemOpen
      \bibfield  {author} {\bibinfo {author} {\bibfnamefont {T.~P.}\ \bibnamefont
      {Lillicrap}}, \bibinfo {author} {\bibfnamefont {J.~J.}\ \bibnamefont {Hunt}},
      \bibinfo {author} {\bibfnamefont {A.}~\bibnamefont {Pritzel}}, \bibinfo
      {author} {\bibfnamefont {N.}~\bibnamefont {Heess}}, \bibinfo {author}
      {\bibfnamefont {T.}~\bibnamefont {Erez}}, \bibinfo {author} {\bibfnamefont
      {Y.}~\bibnamefont {Tassa}}, \bibinfo {author} {\bibfnamefont
      {D.}~\bibnamefont {Silver}}, \ and\ \bibinfo {author} {\bibfnamefont
      {D.}~\bibnamefont {Wierstra}},\ }\href {https://arxiv.org/abs/1509.02971}
      {\bibfield  {journal} {\bibinfo  {journal} {arXiv:1509.02971}}}\BibitemShut {NoStop}%
    \bibitem [{\citenamefont {Leonhardt}(1997)}]{leonhardt1997measuring}%
      \BibitemOpen
      \bibfield  {author} {\bibinfo {author} {\bibfnamefont {U.}~\bibnamefont
      {Leonhardt}},\ }\href@noop {} {\emph {\bibinfo {title} {Measuring the Quantum
      State of Light}}},\ (\bibinfo  {publisher}
      {Cambridge University Press, Cambridge},\ \bibinfo {year} {1997}), Vol.~\bibinfo {volume} {22}\BibitemShut {NoStop}%
    \bibitem [{\citenamefont {James}\ \emph {et~al.}(2001)\citenamefont {James},
      \citenamefont {Kwiat}, \citenamefont {Munro},\ and\ \citenamefont
      {White}}]{James2001measurement}%
      \BibitemOpen
      \bibfield  {author} {\bibinfo {author} {\bibfnamefont {D.~F.~V.}\
      \bibnamefont {James}}, \bibinfo {author} {\bibfnamefont {P.~G.}\ \bibnamefont
      {Kwiat}}, \bibinfo {author} {\bibfnamefont {W.~J.}\ \bibnamefont {Munro}}, \
      and\ \bibinfo {author} {\bibfnamefont {A.~G.}\ \bibnamefont {White}},\ }\href
      {\doibase 10.1103/PhysRevA.64.052312} {\bibfield  {journal} {\bibinfo
      {journal} {Phys. Rev. A}\ }\textbf {\bibinfo {volume} {64}},\ \bibinfo
      {pages} {052312} (\bibinfo {year} {2001})}\BibitemShut {NoStop}%
    \bibitem [{\citenamefont {Vandersypen}\ and\ \citenamefont
      {Chuang}(2005)}]{Vandersypen2005NMR}%
      \BibitemOpen
      \bibfield  {author} {\bibinfo {author} {\bibfnamefont {L.~M.~K.}\
      \bibnamefont {Vandersypen}}\ and\ \bibinfo {author} {\bibfnamefont {I.~L.}\
      \bibnamefont {Chuang}},\ }\href {\doibase 10.1103/RevModPhys.76.1037}
      {\bibfield  {journal} {\bibinfo  {journal} {Rev. Mod. Phys.}\ }\textbf
      {\bibinfo {volume} {76}},\ \bibinfo {pages} {1037} (\bibinfo {year}
      {2005})}\BibitemShut {NoStop}%
    \bibitem [{\citenamefont {Paszke}\ \emph {et~al.}(2019)\citenamefont {Paszke},
      \citenamefont {Gross}, \citenamefont {Massa}, \citenamefont {Lerer},
      \citenamefont {Bradbury}, \citenamefont {Chanan}, \citenamefont {Killeen},
      \citenamefont {Lin}, \citenamefont {Gimelshein}, \citenamefont {Antiga},
      \citenamefont {Desmaison}, \citenamefont {Kopf}, \citenamefont {Yang},
      \citenamefont {DeVito}, \citenamefont {Raison}, \citenamefont {Tejani},
      \citenamefont {Chilamkurthy}, \citenamefont {Steiner}, \citenamefont {Fang},
      \citenamefont {Bai},\ and\ \citenamefont {Chintala}}]{PyTorch2019}%
      \BibitemOpen
      \bibfield  {author} {\bibinfo {author} {\bibfnamefont {A.}~\bibnamefont
      {Paszke}}, \bibinfo {author} {\bibfnamefont {S.}~\bibnamefont {Gross}},
      \bibinfo {author} {\bibfnamefont {F.}~\bibnamefont {Massa}}, \bibinfo
      {author} {\bibfnamefont {A.}~\bibnamefont {Lerer}}, \bibinfo {author}
      {\bibfnamefont {J.}~\bibnamefont {Bradbury}}, \bibinfo {author}
      {\bibfnamefont {G.}~\bibnamefont {Chanan}}, \bibinfo {author} {\bibfnamefont
      {T.}~\bibnamefont {Killeen}}, \bibinfo {author} {\bibfnamefont
      {Z.}~\bibnamefont {Lin}}, \bibinfo {author} {\bibfnamefont {N.}~\bibnamefont
      {Gimelshein}}, \bibinfo {author} {\bibfnamefont {L.}~\bibnamefont {Antiga}}
      \emph{et~al.},\ }in\ \href
      {http://papers.neurips.cc/paper/9015-pytorch-an-imperative-style-high-performance-deep-learning-library.pdf}
      {\emph {\bibinfo {booktitle} {Advances in Neural Information Processing
      Systems 32: Proceedings of the 33rd Conference on Neural Information
      Processing Systems, Vancouver, 2019}}},\ \bibinfo {editor} {edited by\ \bibinfo {editor}
      {\bibfnamefont {H.}~\bibnamefont {Wallach}}, \bibinfo {editor} {\bibfnamefont
      {H.}~\bibnamefont {Larochelle}}, \bibinfo {editor} {\bibfnamefont
      {A.}~\bibnamefont {Beygelzimer}}, \bibinfo {editor} {\bibfnamefont
      {F.}~\bibnamefont {d'~Alch\'{e}-Buc}}, \bibinfo {editor} {\bibfnamefont
      {E.}~\bibnamefont {Fox}}, \ and\ \bibinfo {editor} {\bibfnamefont
      {R.}~\bibnamefont {Garnett}}}\ (\bibinfo  {publisher} {Curran, Red Hook},\ \bibinfo {year} {2019})\ pp.\ \bibinfo {pages}
      {8024--8035}\BibitemShut {NoStop}%
    \bibitem [{\citenamefont {Wu}\ \emph {et~al.}(2020)\citenamefont {Wu},
      \citenamefont {Jin}, \citenamefont {Wen},\ and\ \citenamefont
      {Wang}}]{wu.2020}%
      \BibitemOpen
      \bibfield  {author} {\bibinfo {author} {\bibfnamefont {S.}~\bibnamefont
      {Wu}}, \bibinfo {author} {\bibfnamefont {S.}~\bibnamefont {Jin}}, \bibinfo
      {author} {\bibfnamefont {D.}~\bibnamefont {Wen}}, \ and\ \bibinfo {author}
      {\bibfnamefont {X.}~\bibnamefont {Wang}},\ }\href
      {https://arxiv.org/abs/2012.10711} {\bibfield  {journal} {\bibinfo  {journal}
      {arXiv:2012.10711}}}\BibitemShut
      {NoStop}%
    \bibitem [{\citenamefont {Durkin}(2010)}]{durkin2010preferred}%
      \BibitemOpen
      \bibfield  {author} {\bibinfo {author} {\bibfnamefont {G.~A.}\ \bibnamefont
      {Durkin}},\ }\href
      {https://iopscience.iop.org/article/10.1088/1367-2630/12/2/023010} {\bibfield
       {journal} {\bibinfo  {journal} {New J. Phys.}\ }\textbf {\bibinfo {volume}
      {12}},\ \bibinfo {pages} {023010} (\bibinfo {year} {2010})}\BibitemShut
      {NoStop}%
    \bibitem [{\citenamefont {Degen}\ \emph {et~al.}(2017)\citenamefont {Degen},
      \citenamefont {Reinhard},\ and\ \citenamefont
      {Cappellaro}}]{degen2017sensing}%
      \BibitemOpen
      \bibfield  {author} {\bibinfo {author} {\bibfnamefont {C.~L.}\ \bibnamefont
      {Degen}}, \bibinfo {author} {\bibfnamefont {F.}~\bibnamefont {Reinhard}}, \
      and\ \bibinfo {author} {\bibfnamefont {P.}~\bibnamefont {Cappellaro}},\
      }\href {\doibase 10.1103/RevModPhys.89.035002} {\bibfield  {journal}
      {\bibinfo  {journal} {Rev. Mod. Phys.}\ }\textbf {\bibinfo {volume} {89}},\
      \bibinfo {pages} {035002} (\bibinfo {year} {2017})}\BibitemShut {NoStop}%
    \bibitem [{\citenamefont {Szczykulska}\ \emph {et~al.}(2016)\citenamefont
      {Szczykulska}, \citenamefont {Baumgratz},\ and\ \citenamefont
      {Datta}}]{szczykulska2016multi}%
      \BibitemOpen
      \bibfield  {author} {\bibinfo {author} {\bibfnamefont {M.}~\bibnamefont
      {Szczykulska}}, \bibinfo {author} {\bibfnamefont {T.}~\bibnamefont
      {Baumgratz}}, \ and\ \bibinfo {author} {\bibfnamefont {A.}~\bibnamefont
      {Datta}},\ }\href {https://doi.org/10.1080/23746149.2016.1230476} {\bibfield
      {journal} {\bibinfo  {journal} {Adv. Phys. X}\ }\textbf {\bibinfo {volume}
      {1}},\ \bibinfo {pages} {621} (\bibinfo {year} {2016})}\BibitemShut {NoStop}%
    \end{thebibliography}
%

\end{document}